\documentclass[review]{elsarticle}

\usepackage[margin=2.7cm]{geometry} 
\usepackage{lineno,hyperref}
\modulolinenumbers[5]

\usepackage{textcomp}
\usepackage{array}
\usepackage{listings}
\usepackage{setspace}
\usepackage{mathptmx}
\usepackage[table, svgnames]{xcolor}
\usepackage{colortbl}
\usepackage{graphicx}
\usepackage{amssymb, amsmath}
\usepackage{blkarray, bigstrut} %

\usepackage{nccmath}   
\usepackage{amsthm}    

\usepackage{caption}
\usepackage{subfig}
\usepackage{epsfig}
\usepackage{times}
\usepackage{float}
\usepackage{rotating}
\usepackage{makeidx}
\usepackage{url}
\usepackage{multirow}
\usepackage{booktabs}
\usepackage{tabularx}
\usepackage[toc,title,page]{appendix}

\usepackage[subfigure, titles]{tocloft}
\usepackage{acronym}
\usepackage{datetime}

\DeclareMathOperator*{\argmin}{arg\,min}

\usepackage{algorithm}
\usepackage{lipsum}
\usepackage{algpseudocode}

\newtheorem{theorem}{Theorem}[section]

\newtheorem{lemma}[theorem]{Lemma}

\newtheorem{definition}{Definition}

\journal{Journal of \LaTeX\ Templates}

\begin{document}

\begin{frontmatter}

\title{Multi-variant {COVID-19} model with heterogeneous transmission rates using deep neural networks}



\author[mymainaddress]{K.D. Olumoyin\corref{mycorrespondingauthor}}
\cortext[mycorrespondingauthor]{Corresponding author}
\ead{kayode.olumoyin@mtsu.edu}

\author[mymainaddress]{A.Q.M. Khaliq}

\author[mysecondaryaddress]{K.M. Furati}

\address[mymainaddress]{Department of Mathematical Sciences,
Middle Tennessee State University,  Murfreesboro, TN 37132, USA}
\address[mysecondaryaddress]{Department of Mathematics,
 King Fahd University of Petroleum and Minerals, Dhahran 31261, Saudi Arabia}

%
%
%

\begin{abstract}
Mutating variants of COVID-19 have been reported across many US states since 2021. In the fight against COVID-19, it has become imperative to study the heterogeneity in the time-varying transmission rates for each variant in the presence of pharmaceutical and non-pharmaceutical mitigation measures. We develop a Susceptible-Exposed-Infected-Recovered mathematical model to highlight the differences in the transmission of the B.1.617.2 delta variant and the original SARS-CoV-2. Theoretical results for the well-posedness of the model are discussed. A Deep neural network is utilized and a deep learning algorithm is developed to learn the time-varying heterogeneous transmission rates for each variant. The accuracy of the algorithm for the model is shown using error metrics in the data-driven simulation for COVID-19 variants in the US states of Florida, Alabama, Tennessee, and Missouri.  Short-term forecasting of daily cases is demonstrated using long short term memory neural network and an adaptive neuro-fuzzy inference system.
\end{abstract}


\end{frontmatter}

\linenumbers

\section{Introduction}

COVID-19 was first reported in China in 2019~\citep{WHO2021}, it has since become a global pandemic. 
In recent months, there have been reports of mutating variants of the virus~\citep{Callaway2021}. 
In 2021, the dominant mutant variant of COVID-19 was the B.1.617.2 delta variant~\citep{CDC2021}.
Effort to combat the spread of COVID-19 have included combinations of pharmaceutical (vaccination and hospitalization) and non-pharmaceutical (social distancing, contact tracing, and facial mask) measures. 

Prior to the onset of COVID-19 mutating variants in the US, the progress seen in the data from several states prompted the ease of the various non-pharmaceutical measures. Amid the news that several states had vaccinated over $70\%$ of its population and a few states had vaccinated between $60\% - 70\%$ of its population, vaccination effort began to slow down in many US states.  
As a result, the existence of mutating variants resulted in a resurgence in cases of infections. 
The Center for Disease Control and Prevention (CDC) reported that the dominant variant in the US in 2021 was the B.1.617.2 delta variant. 
According to the World Health Organization (WHO), many variants were first reported in the United Kingdom and South Africa and in recent months, the USA, Europe, China, Brazil, and Japan have all reported mutating variant infected cases. 

We present a data-driven deep learning algorithm for a model consisting of time-varying transmission rates for each active variant. Using infected daily cases data, we learn the form of the time-varying transmission rates, to reveal a timeline of the impact of mitigation measures on the transmission of COVID-19~\citep{olumoyin2021,OlumoyinEINNarxiv}. It can also be demonstrated that this algorithm shows improvement on short-term forecasting when combined with a recurrent neural network and an adaptive neuro-fuzzy inference system.


Neural networks are universal approximators of continuous functions~\citep*{Cybenko1989, Hornik1991}.
Feedforward neural networks (FNN) have been used to learn approximate solutions of differential equations.
In~\citep{Raissi2019}, FNN was used to develop differential equation solvers and parameter estimators by constraining the residual.
This FNN is called the Physics Informed Neural Network (PINN).
PINN has been used to simulate pandemic spread, see~\citep{RaisiM2019}, where the model parameters were taken to be constants.
In~\citep{Long2020}, an algorithm that combines PINN with Long Short-term Memory (LSTM) is presented to solve an epidemiological model and identify weekly and daily time-varying parameters.

The paper is organized as follows. 
In Section \ref{model}, we introduce and discuss the multi-variant SEIR model and the time-varying transmission rate of each variant. The well posedness of the model is discussed in Section~\ref{well}. The neural network structure of the Epidemiology neural network EINN is presented in~\ref{EINN}. Data-driven simulation of COVID-19 data is shown in Sections~\ref{data1}. A comparison of a recurrent neural network based forecast and an adaptive neuro-fuzzy inference system based forecast is presented in~\ref{data2}. The performance error metrics of EINN is discussed in Section~\ref{metrics}. The paper is summarized in Section~\ref{conclusion}.


\section{Multi-variant SEIR model }
\label{model}

We assume that the total population $N(t)=N$ at any given time is distributed among the following compartments: 
susceptible $(S)$, exposed $(E)$, Infectious $(I_i)$, $i=1,\ldots,M$, and recovered $(R)$, where $M$ is the number of different variants. The interaction between the compartments is shown in Figure~\ref{trans_diag}.

\begin{figure}[htbp]
\centering
\includegraphics[width=14cm]{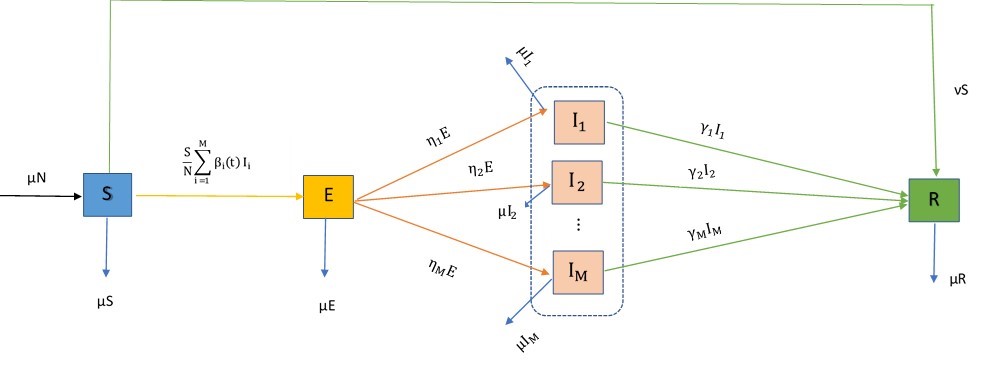}
\caption{Transfer diagram between the compartments} \label{trans_diag}
\end{figure}

As shown in Figure~\eqref{trans_diag}, the susceptible individuals enter the exposed compartment at the rate $\frac{1}{N}\sum_{i=1}^{M}\beta_i(t)I_i$, where $\beta_i(t)$ is the transmission rate of variant $i$. The exposed individuals progress to the $i$th infected compartment at the rate $\eta_i$. The $i$th infected compartment recover at the rate of $\gamma_i I_i$.

We assume a natural death rate $\mu$, given by $\mu = \frac{1}{\it{LFE} \times 365}$, where $\it{LFE}$ denotes the life expectancy. For simplicity, it is assumed that the birth rate of the population is  equal to the death rate. The parameter $\eta$ is the transmission rate from the exposed to the various infectious sub-compartments, $\gamma_i^{-1}$ is the mean symptomatic infectious period for the $i$th variant. The parameter $\nu(t)$ represent the time-dependent removal rate of the vaccinated individuals from the susceptible compartment. We assume that any variant does not super-infect another variant, so there are no interactions between the infectious sub-compartments.

Based on the transfer diagram depicted in Figure~\ref{trans_diag}, the mathematical model for a multi-variant COVID-19 pandemic with heterogeneous transmission rates is given by:

\begin{equation}\label{A-SIReq} 
\begin{split}
\frac{d S}{dt}   & = \mu N -\frac{S}{N} \sum_{i=1}^{M}\beta_i(t) I_i  -  \Big(\nu(t) + \mu\Big) S \\
\frac{d E}{dt}  & =  \frac{S}{N} \sum_{i=1}^{M}\beta_i(t) I_i  - (\eta + \mu) E \\
\frac{d I_i}{dt} & = \eta_i E - (\gamma_i + \mu) I_i, \hspace{5 MM} i=1,\ldots,M \\
\frac{d R}{dt} & =  \sum_{i=1}^{M}\gamma_i I_i + \nu(t) S - \mu R
\end{split}
\end{equation}

subject to non-negative initial conditions
\[ S(0) = S_0, \quad E(0) = E_0, \quad I_i(0) = I_{i,0}, i=i,\dots,M, \quad R(0) = R_0. \]

The parameter $\eta$ is defined as: $\eta = \sum_{i=1}^{M} \eta_i$, and the total population is 
\[ N(t) = S(t) + E(t) + \sum_{i=1}^{M} I_i(t) + R(t). \]

The differential equation satisfied by the total population size is obtained by adding all the equations in  \eqref{A-SIReq}, that is, $\frac{d N}{dt}=0$ and thus N is constant. The model parameters are summarized in Table~\ref{Table:10}.

Time-varying transmission rates have been shown to efficiently model the spread of COVID-19~\citep{olumoyin2021,Jagan2020}. Next, will discuss the form of the time-varying transmission rates for each variant.

\subsection{Variant-based time-varying transmission rates}
\label{time-vary}

Time-varying transmission rates  in~\eqref{A-SIReq} incorporates the impact of governmental actions, and the public response \cite{HeD2013}. We consider the transmission rates of the form

\begin{equation}\label{webb_eqn2}
\beta_i(t) = \beta_{i}^{0} \exp{(-\kappa_i t)},  \hspace{5 MM} 1 \leq i \leq M.
\end{equation}
where $\kappa_i$ in~\eqref{webb_eqn2}  is the infectiousness factor for each $i$ th variant. However, we define $(1+\tau_i)$ to be the factor by which a particular variant is more infectious than the original variant SARS-CoV-2. And so, the following relationship exist between each mutating variant and the SARS-CoV-2 variant~\eqref{webb_eqn3}.

\begin{equation}\label{webb_eqn3}
\beta_i(t) = \beta_1(t) (1 + \tau_i),  \hspace{5 MM} 2 \leq i \leq M.
\end{equation}

In~\eqref{webb_eqn3}, $ \beta_1(t)$ is the transmission rate for the original variant SARS-CoV-2.
The transmission rates of the subsequent mutating variants are given by  $\beta_i(t)$, $i=2,\dots,M$, where M represents the number of mutating COVID-19 variants. Although the publicly available data reports the daily infected cases, there were reports that suggest hat the dominant variant, B.1.617.2 delta variant, to be twice as infectious as the original variant SARS-CoV-2. According to CDC reports~\citep*{CDC2021,CDC2021_2}, the delta variant accounted for $1.3\%$ of total infected cases in May, 2021, $9.5\%$ in June, and in August it accounted for $93\%$ of the total infected cases. 

\begin{table}[H]
\begin{tabular}{lcccc}
\addlinespace
\toprule
{\bf Parameter} & {\bf Notation} & {\bf Range} & {\bf Remark} & {\bf Reference}\\
\toprule
{Baseline transmission rate for each $i$th variant} & {$\beta_{i}^{0}$} & {[0,1)} & {fitted using daily cases data} & {\cite{olumoyin2021}}  \\
{Emigration rate} & {$\mu$} & {$\frac{1}{\it{LFE} \times 365}$} & {constant}  & {\cite{Furati2021}}   \\ 
{Mean latent period} & ${\eta^{-1}}$ & {$2 - 14 (days)$} & {constant}  & {\cite{Furati2021}}   \\ 
{recovery rate for each $i$th variant} & ${\gamma_{i}}$ & {$[0,1)$} & {constant}  &  \\
{infectiousness factor for each $i$th variant} & ${\kappa_{i}}$ & {$[0,1)$} & {constant}  &  \\ 
\midrule
\end{tabular}
\caption{Summary table of parameters in model~\eqref{A-SIReq}}
\label{Table:10}
\end{table}

\subsection{Well-posedness of the model}
\label{well}

\begin{definition}[\cite{Schaeffer2016}, Locally Lipschitz continuity]
Let $d_1, d_2 \in \mathbf{N}$ and $\mathbf{S}$ be a subset of $\mathbf{R}^{d_1}$. A function $\mathbf{F}:\mathbf{S} \rightarrow \mathbf{R}^{d_2}$ is Lipschitz continuous on $\mathbf{S}$ if there exists a nonnegative constant $L\geq 0$ such that 

\begin{equation}\label{lipschitz}
|\mathbf{F}(x) - \mathbf{F}(y)| \leq L|x-y|, \quad x,y \in \mathbf{S}.
\end{equation}

\end{definition}

Let $\mathcal{U}$ be an open subset of $\mathbf{R}^{d_1}$, and let $\mathbf{F}:\mathcal{U} \rightarrow \mathbf{R}^{d_2}$. We shall call $\mathbf{F}$ locally Lipschitz continuous if for every point $x_0 \in \mathcal{U}$ there exists a neighborhood $V$ of $x_0$ such that the restriction of $\mathbf{F}$ to $V$ is Lipschitz continuous on $V$.

We consider a more general framework of model \eqref{A-SIReq} 
\begin{equation}\label{ivpp}
z^{\prime}(t)=G(z(t)), \hspace{5 MM} z(0)=z_0,
\end{equation}
 where $z(t)=(x_1(t),x_2(t),\ldots,x_n(t))^T$ and $G(z(t)) = (g_1(z(t)),g_2(z(t)),\ldots,g_n(z(t)))^T$, the initial condition $z_0 \in \mathbf{R}^n$. We state the following theorem.

\begin{theorem}[\cite{Schaeffer2016}]\label{sch}
If $\mathbf{G}:\mathbf{R}^{n} \rightarrow \mathbf{R}^{n}$ is locally Lipschitz continuous and if there exist nonnegative constants $B$, $K$ such that

\begin{equation}
|\mathbf{G}(z(t))|\leq K \hspace{1 MM}  |z(t)| + B, \quad z(t) \in \mathbf{R}^n,
\end{equation}

then the solution of the initial value problem \eqref{ivpp} exists for $t >0$, and 

\begin{equation}
|(z(t))|\leq |z_0|\cdot \exp{(K \hspace{1 MM} |t|)} + \frac{B}{K} \cdot (\exp{(K \cdot|t|)}-1), \quad t >0,
\end{equation}

\end{theorem}

\begin{lemma}\label{bdd}
For each $i$th variant, $i \in \{1, \ldots, M\}$, the time varying transmission rates $\nu,\beta_i: [0,\infty) \rightarrow [0,\infty)$ are Lipschitz continuous and contnuously differentiable. There exists $\beta_{min},\beta_{max}>0$ and $\nu_{min} \nu_{max}>0$ such that $\beta_{min} \leq \beta_i(t) \leq \beta_{max}$, $\nu_{min}\leq \nu(t) \leq \nu_{max}$ for all $t$.
\end{lemma}

\begin{theorem}\label{exist}
The nonlinear first order system of differential equations \eqref{A-SIReq} has at least one solution which exists for $t \in [0, \infty)$.
\end{theorem}
\begin{proof}
Let $z(t)=(S(t),E(t),I_1(t),\ldots,I_M(t),R(t))^T$, we can set 

\begin{equation}\label{GG}
G: \mathbf{R}^{M+3} \rightarrow \mathbf{R}^{M+3},
\hspace{10MM}
z(t) \rightarrow \begin{pmatrix}
\mu N -\frac{S}{N} \sum_{i=1}^{M}\beta_i(t)I_i  -  \Big(\nu(t) + \mu\Big) S \\
\frac{S}{N} \sum_{i=1}^{M}\beta_i(t)I_i  - (\eta + \mu) E \\
\eta_i E - (\gamma_i + \mu) I_i, \hspace{5 MM} i=1,\ldots,M \\
\sum_{i=1}^{M}\gamma_i I_i + \nu(t) S - \mu R 
\end{pmatrix}
\end{equation}

$G$ is locally Lipschitz continuous, using supremum norm $||f(t)|| := \sup \limits_{t \in [a,b]}|f(t)|$, we have
\begin{equation*}
\begin{split}
||G(z(t))|| &=  \sup \limits_{t \in [0,\infty)} \Bigg\{\Bigg|\mu N -\frac{S(t)}{N} \sum_{i=1}^{M}\beta_i(t)I_i(t)   -  \Big(\nu(t) + \mu\Big) S(t) \Bigg|, \Bigg|\frac{S(t)}{N} \sum_{i=1}^{M}\beta_i(t)I_i - (\eta + \mu) E(t) \Bigg|, \\
& \qquad \Big| \eta_i E(t) - (\gamma_i + \mu) I_i(t), \hspace{5 MM} i=1,\ldots,M \Big|,
\Big| \sum_{i=1}^{M}\gamma_i I_i(t) + \nu(t) S(t) - \mu R(t) \Big|\Bigg\}\\
& \leq  \sup \limits_{t \in [0,\infty)}\Bigg\{\mu N +  \beta_{max}\Bigg|\frac{S(t)}{N} \sum_{i=1}^{M}I_i(t)\Bigg|   +  (\nu_{max}+\mu)| S(t) |, \beta_{max}\Bigg|\frac{S(t)}{N} \sum_{i=1}^{M}I_i(t)\Bigg| +  (\eta + \mu)| E(t) |, \\
& \qquad  \eta_i |E(t)|  + (\gamma_i + \mu) |I_i(t)|,  \gamma_i \sum_{i=1}^{M}| I_i(t)| + \nu_{max} |S(t)| + \mu |R(t)| \Bigg\}\\
& \leq  \sup \limits_{t \in [0,\infty)}\Big\{\mu N +  \beta_{max}|S(t)|+(\nu_{max}+\mu)| S(t) |, \beta_{max}|S(t)|+  (\eta + \mu)| E(t) |, \\ 
& \qquad  \eta_i |E(t)|  + (\gamma_i + \mu) |I_i(t)|, \gamma_i \sum_{i=1}^{M} |I_i(t)| + \nu_{max} |S(t)| + \mu |R(t)| \Big\}\\
& \leq K ||z(t)||.
\end{split}
\end{equation*} 
So by Theorem~\ref{sch}, and the boundedness of the time-varying nonlinear functions from Lemma~\ref{bdd}, the nonlinear initial value problem~\ref{A-SIReq} has a solution for all time.
\end{proof}

\subsection{Basic reproduction number and equilibria stability}

The basic reproduction number $\mathcal{R}_0$ is the expected number of secondary infections that a single infectious individual will generate on average within a susceptible population. 

\begin{definition}\label{dfe}
The disease-free equilibrium of \eqref{A-SIReq} is given by 
\[(S^{\ast}, E^{\ast}, I_1^{\ast}, \ldots, I_M^{\ast}, R^{\ast}) = (S_0, 0, 0, \ldots, 0, 0).\]
\end{definition}
The basic reproduction number $\mathcal{R}_0$  is calculated for the case when $\beta_i(t)=\beta_{i}^{0}$, $i \in \{1, \ldots, M\}$.
Applying the next-generation operator approach~\citep*{Driessche2002}, the reproduction number $\mathcal{R}_0$ is obtained as the spectral radius of the next generation matrix $FV^{-1}$, where
\[
F = \begin{pmatrix}
 0 & \beta_{1}^{0} & \ldots & \beta_{M}^{0} \\
0 & 0 & \ldots & 0\\
\vdots & \vdots & \vdots & \vdots\\
0 & 0 & \ldots & 0
\end{pmatrix},
\hspace{5 MM}
V = \begin{pmatrix}
 \eta + \mu & 0 & 0 & \ldots & 0\\
-\eta_1 & \gamma_1 + \mu & 0 & \ldots & 0\\
-\eta_2 & 0 & \gamma_2 + \mu &  \ldots & 0\\
\vdots & \vdots & \vdots & \vdots & \vdots\\
-\eta_M & 0 & 0 & \ldots & \gamma_M + \mu
\end{pmatrix}.
\]

The basic reproduction number $\mathcal{R}_0$ is computed as follows in ~\eqref{chap4_r0}

\begin{equation}\label{chap4_r0}
\mathcal{R}_0 =  \sum_{i=1}^{M}\frac{\beta_{i}^{0}\eta_i}{ \left(\gamma_i + \mu \right)\left(\eta + \mu \right)}.
\end{equation}

Next, we analyze the local asymptotic stability of the disease-free equilibrium in Definition~\ref{dfe}.

\begin{theorem}\label{thm1}
The disease-free equilibrium $(S^{\ast}, E^{\ast}, I_1^{\ast}, \ldots, I_M^{\ast}, R^{\ast})$ of \eqref{A-SIReq} is locally asymptotically stable if $\mathcal{R}_0 < 1$. 
\end{theorem}
\begin{proof}
The Jacobian of the right hand side of \eqref{A-SIReq} at the equilibrium point is given by
\[
J = \begin{pmatrix}
 -(\nu+\mu) & 0 & - \beta_{1}^{0} & - \beta_{2}^{0} & \ldots & - \beta_{M}^{0} & 0\\
 0 & -(\eta +\mu) &  \beta_{1}^{0} &  \beta_{2}^{0} & \ldots &  \beta_{M}^{0} & 0\\
 0 & \eta_1 & -(\gamma_1 + \mu) & 0 & \ldots & 0 & 0\\
 0 & \eta_2 & 0 & -(\gamma_2 + \mu) & \ldots & 0 & 0\\
 \vdots & \vdots & \vdots & \vdots & \vdots & \vdots & \vdots\\
 0 & \eta_M & 0 & 0 & \ldots & -(\gamma_M + \mu) & 0\\
 \nu & 0 & \gamma_1 & \gamma_2 & \ldots & \gamma_M & -\mu
\end{pmatrix}
\]
If $M=1$, the eigenvalues of the Jacobian matrix are given as follows:
\[
\lambda_1 =  -\mu,
\lambda_2 = -(\nu+\mu),
\lambda_3 = \frac{1}{2} \Big(-A - B \Big),
\lambda_4 = \frac{1}{2} \Big(A - B \Big).
\]
where
\begin{equation*}
\begin{split}
A &=  \sqrt{4\beta_1^0\eta_1 + (\eta_1 - \gamma_1)^2} ,\\
B &= \eta_1 + 2\mu + \gamma_1.
\end{split}
\end{equation*}
Clearly, $A,B>0$ and $A<B$, so that $\lambda_1,\lambda_2,\lambda_3,\lambda_4 < 0$.
Similarly, we can show negative eigenvalues for $M\geq2$. So the disease-free equilibrium is locally asymptotically stable. 
\end{proof}

\section{Epidemiology Informed Neural Network (EINN)} 
\label{EINN}

A Feedforward Neural Network (FNN) composed of $L$ layers, $t$ inputs and an output $\mathcal{N}$ can be represented as the following function 
\begin{equation}
\mathcal{N}(t;\Sigma) = \sigma(W_{L}\sigma(\ldots \sigma(W_{2} \sigma(W_{1}t + b_1)+b_2)\ldots)+b_L),
\end{equation}
\noindent where $\Sigma:=(W_1, \ldots, W_L,b_1, \ldots, b_L)$. The neural network weight matrices are $W_l$, $l = 1, \ldots, L$, while the bias vectors are  $b_l$, $l = 1, \ldots, L$. Here, $\sigma$ is the activation function. Given a collection of sample pairs $(t_k, u_k)$, $k = 1,\dots K$, where $u$ is some target function, the goal is to find $\Sigma^{*}$ by solving the following optimization problem
\begin{equation}\label{orig_lossFunc}
\Sigma^{*} = \argmin\limits\frac{1}{K}\sum_{k=1}^{K}||\mathcal{N}(t_k;\Sigma) - u_k||_2^2.
\end{equation}
\noindent The function $\frac{1}{K}\sum_{k=1}^{K}||\mathcal{N}(t_k;\Sigma) - u_k||_2^2$ on the right hand side of \eqref{orig_lossFunc} is called the mean squared error (MSE) loss function.
A major task in training a network is to determine the suitable number of layers and the number of neurons per layer needed, the choice of activation function, and an appropriate optimizer for the loss function \cite{Goodfellow2016}.

\begin{figure}[htbp]
\centering
\includegraphics[width=15cm]{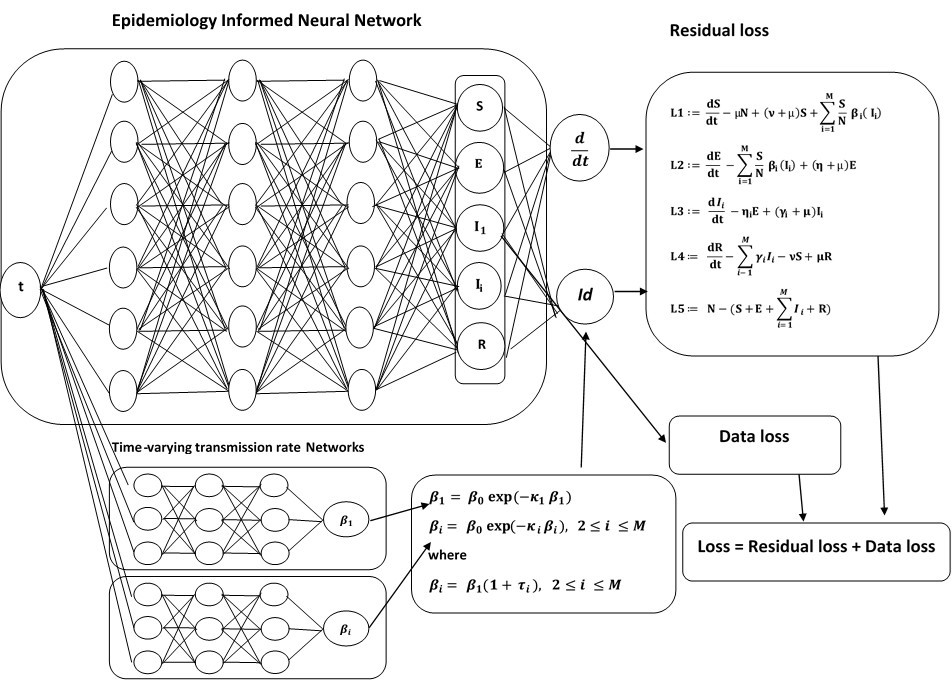}
\caption{Schematic diagram of the Epidemiology Informed Neural Network with nonlinear time-varying transmission rate.  \label{Sch_A_SEIR}}
\end{figure}

EINN is a form of Feedforward Neural Network that includes the known epidemiology dynamics in its loss function.
EINN is adapted for the SEIR model \eqref{A-SIReq}, where the Mean Square Error (MSE) of this neural network's loss function includes the known epidemiology dynamics such as a lockdown, while other mitigation measures such as social distancing, and contact tracing are detected by the time-varying transmission rate.
The output of EINN are the learned solutions to the SEIR model \eqref{A-SIReq} denoted by $S(t_j; \psi; \rho)$, $E(t_j; \psi; \rho)$, $I(t_j; \psi; \rho)$, $R(t_j; \psi; \rho)$, $j=1,\ldots,T$. 
Here, $\psi$ represent the neural network weights and biases while $\rho$ represent the epidemiology parameters and $T$ is the number of days in our dataset.
Next, we set-up time-varying transmission rate networks whose outputs are $\beta_i(t_j;\pi_i;\kappa_i)$, $j=1,\ldots,T$, for $i=1,\dots,M$. Each $\pi_i$ represent the weights and biases of each $i$th network and $\kappa_i$ is the infectiousness factor for each $i$th variant.
The training data is generated using cubicspline and denoted by $\tilde{I}(t_j)$ and $\tilde{V}(t_j)$, $j=T_{\nu},\ldots,T$. 
Here $T_{\nu}$ is an integer that correspond to the vaccination start date in the dataset. The B.1.617.2 delta variant was first reported in the USA in May, $T_{\delta}$ is an integer that correspond to May $4th$, 2021.
We observe that training data is not available for all the compartments in the SEIR model, however, EINN is able to capture the epidemiology interactions between the compartments because the residual of equation~\eqref{A-SIReq} is included in the MSE loss function.

The MSE loss function for EINN is given by,

\begin{equation}\label{loss_func}
\begin{split}
MSE &= \frac{1}{T_{\delta}}\sum_{j=1}^{T_{\delta}}\Big|I_1(t_j; \psi; \rho) - \tilde{I}(t_j)\Big|^2 
+ \frac{1}{|T-T_{\delta}|}\sum_{j=T_{\delta}}^{T}\Big|\sum_{i=1}^{M}I_i(t_j; \psi; \rho) - \tilde{I}(t_j)\Big|^2 \\
& +  \frac{1}{T_{\nu}}\sum_{j=1}^{T_{\nu}}\Big|\nu(t_j;\psi; \rho) - 0\Big|^2
+ \frac{1}{|T-T_{\nu}|}\sum_{j=T_{\nu}}^{T}\Big|\nu(t_j;\psi; \rho) S(t_j;\psi; \rho) - \tilde{V}(t_j)\Big|^2\\
& + \sum_{i=2}^{M}\Big|I_i(t_{\delta_{1}}; \psi; \rho) - p_{\delta_{1}} \tilde{I}(t_{\delta_{1}})\Big|^2
+ \sum_{i=2}^{M}\Big|I_i(t_{\delta_{2}}; \psi; \rho) - p_{\delta_{2}} \tilde{I}(t_{\delta_{2}})\Big|^2\\
& + \frac{1}{T_{\delta}}\sum_{i=2}^{M}\sum_{j=1}^{T_{\delta}}\Big|\beta_i(t_j; \pi_i;\kappa_i) - 0\Big|^2\\
& +  \frac{1}{|T-T_{\delta}|}\sum_{i=2}^{M}\sum_{j=T_{\delta}}^{T}\Big|\beta_1(t_j; \pi_1;\kappa_1)(1 + \tau_i) - \beta_i(t_j; \pi_i;\kappa_i)\Big|\\
& + \sum_{l=1}^5 L_l,
\end{split}
\end{equation}
where the residual $L_l$, $l=1,\ldots 5$, is as follows
\begin{equation}\label{resid_loss}
\begin{split}
L_1 &=\frac{1}{T}\sum_{j=1}^{T}\Big|\frac{d S(t_j; \psi;\rho)}{d t_j} + \frac{S(t_j; \psi;\rho)}{N}\Big(\sum_{i=1}^{M} \beta_i(t_j; \pi_i; \kappa_i) I_i(t_j; \psi;\rho) \Big) + \Big(\nu(t_j;\psi; \rho) + \mu\Big) S(t_j; \psi;\rho) - \mu N\Big|^2 \\
L_2 &=\frac{1}{T}\sum_{j=1}^{T}\Big|\frac{d E(t_j; \psi;\rho)}{d t_j} - \frac{S(t_j; \psi;\rho)}{N}\Big(\sum_{i=1}^{M} \beta_i(t_j; \pi_i; \kappa_i) I_i(t_j; \psi;\rho) \Big) + (\eta+\mu) E(t_j; \psi;\rho)\Big|^2\\
L_3 &= \frac{1}{T}\sum_{i=1}^{M}\sum_{j=1}^{T}\Big|\frac{d I_i(t_j; \psi;\rho)}{d t_j} - \eta_i E(t_j; \psi;\rho) + (\gamma_i + \mu) I_i(t_j; \psi;\rho)\Big|^2\\
L_4 &=\frac{1}{T}\sum_{j=1}^{T}\Big|\frac{d R(t_j; \psi;\rho)}{d t_j}   - \sum_{i=1}^{M}\gamma_i I_i(t_j; \psi;\rho) - \nu(t_j;\psi; \rho) S(t_j; \psi;\rho) + \mu R(t_j;\psi;\rho)\Big|^2 \\
L_5 &=  \frac{1}{T}\sum_{j=1}^{T}\Big|N - (S(t_j; \psi;\rho) + E(t_j; \psi;\rho) + \sum_{i=1}^{M}I_i(t_j; \psi;\rho)  + R(t_j; \psi;\rho))\Big|^2.
\end{split}
\end{equation}
where $\eta = \sum_{i=1}^{M}\eta_i$,  $i=1,\ldots M$.

The daily infected cases, the vaccinated cases, the known COVID-19 variants facts and the transmission rates are enforced in the mean square error (MSE)~\eqref{loss_func}, see Figure~\eqref{Sch_A_SEIR}. For instance, $p_{\delta_{1}}$ and $p_{\delta_{2}}$ correspond to the proportion of daily cases that was due to the mutating variants as reported by the CDC~\citep{CDC2021_2}.  


\begin{algorithm}[H]
\begin{algorithmic}[1]
\footnotesize

\State {Construct EINN}
 
{specify the input: $t_j$, $j=1,\ldots,T$}

{Initialize EINN parameter: $\psi$}

{Initialize the mathematical model parameters: $\rho=[\gamma_i]$, $i=1,\ldots,M$.}

{Output layer: $S(t_j; \psi; \rho)$, $E(t_j; \psi; \rho)$, $I(t_j; \psi; \rho)$, $R(t_j; \psi; \rho)$, $j=1,\ldots,T$}

\State{construct neural networks: $\beta_i$, $j=1,\ldots,M$}

{specify the input: $t_j$, $j=1,\ldots,T$}

{Initialize the neural network parameter: $\phi$}

{Specify $\beta_i^0$ obtained by fitting daily cases}

{Output layers : $\beta_i(t_j;\pi_i;\kappa_i)$  

\begin{equation}
\beta_i(t_j;\pi_i;\kappa_i) = (1+\tau_i)\beta_1(t_j;\pi_1;\kappa_1),  \hspace{5 MM} i\geq 2
\end{equation}}

\State{Specify EINN training set }

{Training data: using cubicspline, generate $\tilde{I}(t_j)$ and $\tilde{R}(t_j)$, $j=1,\ldots,T$.}


\State{Train the neural networks}

{Specify an $MSE$ loss function:
\begin{equation}\label{loss_func_beta_2}
 \begin{split}
MSE &= \frac{1}{T_{\delta}}\sum_{j=1}^{T_{\delta}}\Big|I_1(t_j; \psi; \rho) - \tilde{I}(t_j)\Big|^2 
+ \frac{1}{|T-T_{\delta}|}\sum_{j=T_{\delta}}^{T}\Big|\sum_{i=1}^{M}I_i(t_j; \psi; \rho) - \tilde{I}(t_j)\Big|^2 \\
& +  \frac{1}{T_{\nu}}\sum_{j=1}^{T_{\nu}}\Big|\nu(t_j;\psi; \rho) - 0\Big|^2
+ \frac{1}{|T-T_{\nu}|}\sum_{j=T_{\nu}}^{T}\Big|\nu(t_j;\psi; \rho) S(t_j;\psi; \rho) - \tilde{V}(t_j)\Big|^2\\
& + \sum_{i=2}^{M}\Big|I_i(t_{\delta_{1}}; \psi; \rho) - p_{\delta_{1}} \tilde{I}(t_{\delta_{1}})\Big|^2
+ \sum_{i=2}^{M}\Big|I_i(t_{\delta_{2}}; \psi; \rho) - p_{\delta_{2}} \tilde{I}(t_{\delta_{2}})\Big|^2\\
& + \frac{1}{T_{\delta}}\sum_{i=2}^{M}\sum_{j=1}^{T_{\delta}}\Big|\beta_i(t_j; \pi_i;\kappa_i) - 0\Big|^2\\
& +  \frac{1}{|T-T_{\delta}|}\sum_{i=2}^{M}\sum_{j=T_{\delta}}^{T}\Big|\beta_1(t_j; \pi_1;\kappa_1)(1 + \tau_i) - \beta_i(t_j; \pi_i;\kappa_i)\Big|\\
& + \sum_{l=1}^5 L_l,
 \end{split}
\end{equation} }

{Minimize the $MSE$ loss function:
compute $\argmin\limits_{\{\psi; \pi_i \}}(MSE)$ using an optimizer such as the $L-BFGS-B$.}

\State {return EINN solution}

{$S(t_j;\psi;\rho)$, $E(t_j;\psi;\rho)$, $I_i(t_j;\psi;\rho)$, $R(t_j;\psi;\rho)$, $j=1,\ldots,T$, $i=1,\ldots,M$.}

{epidemiology parameters: $\gamma_i$, $i=1,\ldots,M$.}

{vaccination parameter: $\nu$ }

\State {return time-varying epidemiology parameter:}

{$\beta_i(t_j;\pi_i;\kappa_i)$, $j=1,\ldots,T$, $i=1,\ldots,M$. }

{Infectiousness factor: $\kappa_i$, $i=1,\ldots,M$. }

{parameter: $\tau_i$, $i=2,\ldots,M$. }

\end{algorithmic}
\caption{EINN algorithm for SEIR model with time-varying transmission rate \label{alg:Epidemiology informed neural network_var}}
\end{algorithm}

\section{Data-driven simulation of COVID-19 variants}
\label{data1}
We present results of the implementation of the EINN algorithm in Figure~\eqref{Sch_A_SEIR} for COVID-19 data from Alabama, Missouri, Tennessee, and Florida. We consider data from March 2020 to September 2021, during which there were two dominant variants;the original variant SARS-CoV-2 and the delta variant (B.1.617.2). CDC report indicate that $1.3\%$ of the total infected cases were due to the delta variant in May 4th 2021~\citep*{CDC2021_2}. The EINN algorithm learns the infected cases, and the time-varying transmission rates due to each variant. In Table~\eqref{Table:al1}--\eqref{Table:fl1}, pre-$\gamma_1$, post-$\gamma_1$, post-$\gamma_2$ denote the recovery rate of people infected due to the original variant SARS-CoV-2 before the onset of the delta variant, recovery rate of people infected due to the original variant SARS-CoV-2 after the onset of the delta variant, and recovery rate of people infected due to the delta variant after the onset of the delta variant respectively.

The CDC reports that by July 31st, 2021, the proportions of infected cases that are due to the B.1.617.2 delta variant in Alabama was $82.6\%$, Tennessee was $67.4\%$, Missouri $53.9\%$, and in Florida, it was $86.4\%$ \cite{CDC2021}. The CDC also reported that in the USA, the delta variant accounted for about $1.3\%$ of the infected cases.

We seek to learn $\tau_i$ for an $i$th mutating variant. For the simulations in this section, we observed that the delta variant is a dominant mutating variant therefore we included only two variants, the SARS-CoV-2 and the delta variant.  

\begin{table}[htbp]
\centering
\begin{tabular}{lcc}
\addlinespace
\toprule
{Parameters} & {Mean} & {Std} \\
\toprule
pre- $\gamma_1$  & 0.02423 & 0.01266 \\
post- $\gamma_1$  & 0.00395 & 0.00717 \\
post- $\gamma_2$  & 0.00463 & 0.00768 \\
$\eta_1$  & 0.12437 & 0.04933 \\
$\eta_2$  & 0.20893 & 0.04933 \\
$\kappa_1$  & 1.07385 & 0.06271 \\
$\kappa_2$  & 1.13052 & 0.02809 \\
$(1 + \tau)$  & 1.22391 & 0.10176 \\
\midrule
\end{tabular}
\caption{Using Alabama daily cases from March 2020 to September 2021, the EINN Algorithm~\eqref{alg:Epidemiology informed neural network_var} learns the model parameters} 
\label{Table:al1}
\end{table}

\begin{figure}[htbp]
\centering
\includegraphics[width=17cm]{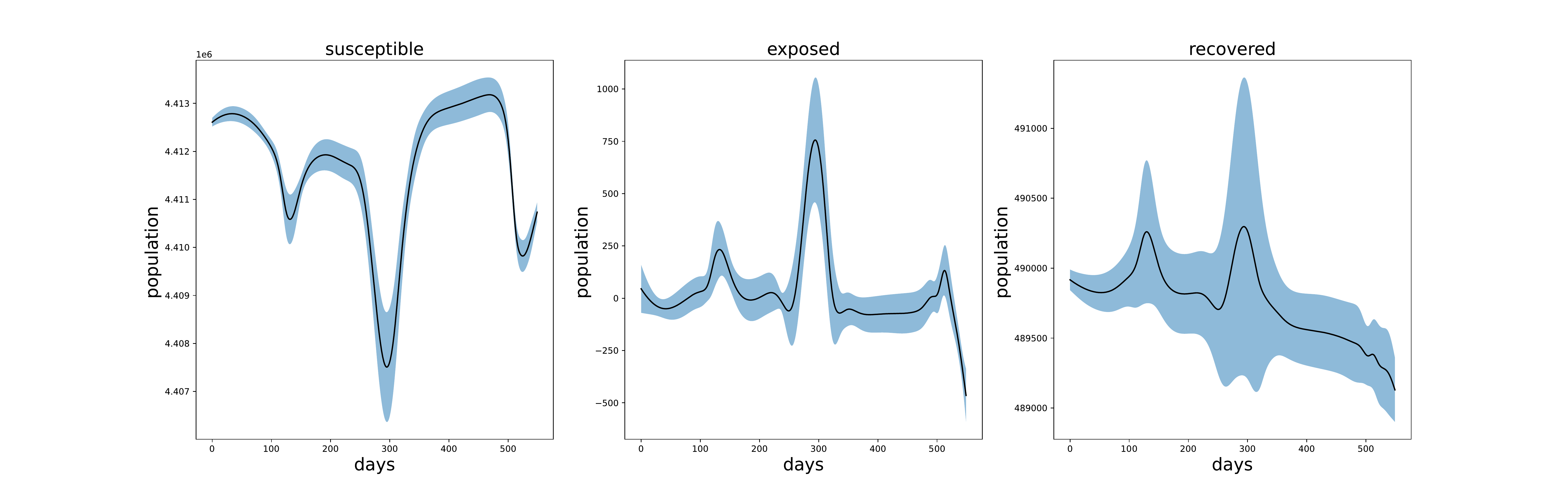}
\caption{learned Alabama Susceptible, Exposed, and Recovered daily population}
\label{alabama_4.1}
\end{figure}

\begin{figure}[htbp]
\centering
\includegraphics[width=17cm]{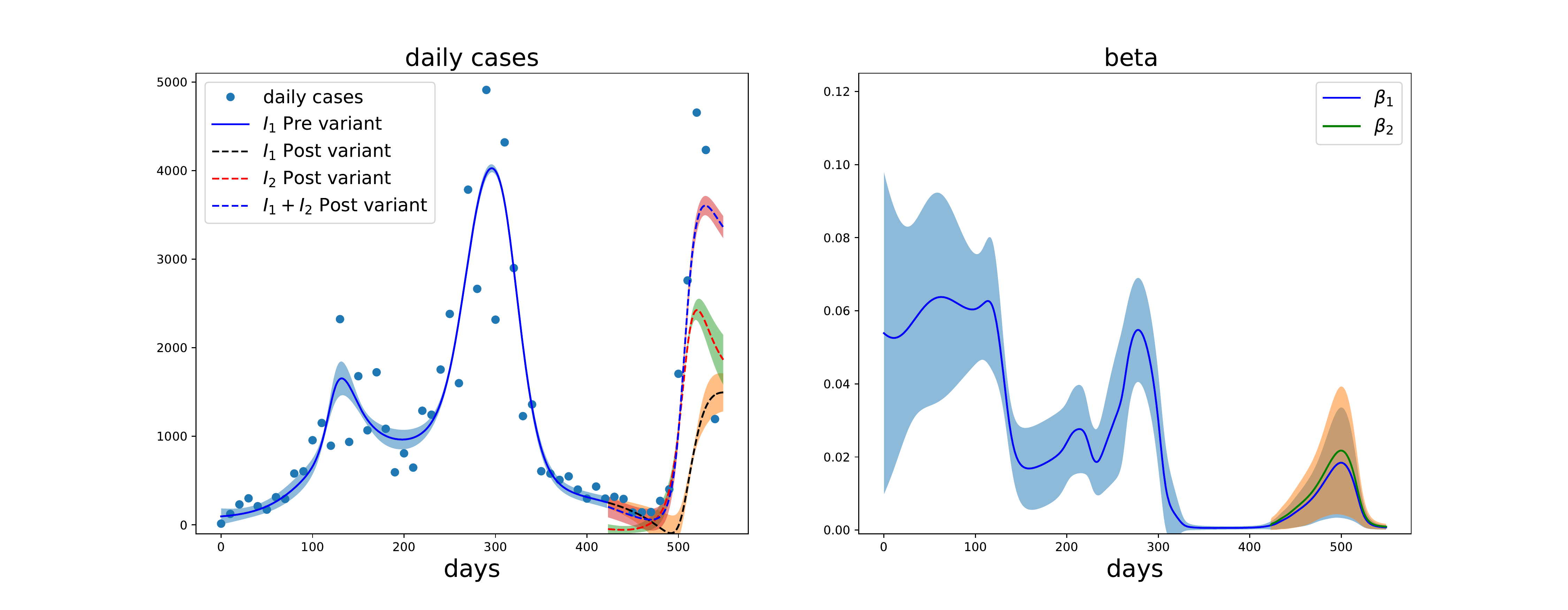}
\caption{Alabama daily cases and time-varying transmission rates}
\label{alabama_4}
\end{figure}


\begin{table}[htbp]
\centering
\begin{tabular}{lcc}
\addlinespace
\toprule
{Parameters} & {Mean} & {Std} \\
\toprule
pre- $\gamma_1$  & 0.02344 & 0.00613 \\
post- $\gamma_1$  & 0.00911 & 0.00426 \\
post- $\gamma_2$  & 0.02095 & 0.01794 \\
$\eta_1$  & 0.15912 & 0.03381 \\
$\eta_2$  & 0.17420 & 0.03383 \\
$\kappa_1$  & 1.01114 & 0.03561 \\
$\kappa_2$  & 1.10474 & 0.01358 \\
$(1 + \tau)$  & 1.15537 & 0.08817 \\
\midrule
\end{tabular}
\caption{Using Missouri daily cases from March 2020 to September 2021, the EINN Algorithm~\eqref{alg:Epidemiology informed neural network_var} learns the model parameters} 
\label{Table:mo1}
\end{table}

\begin{figure}[htbp]
\centering
\includegraphics[width=17cm]{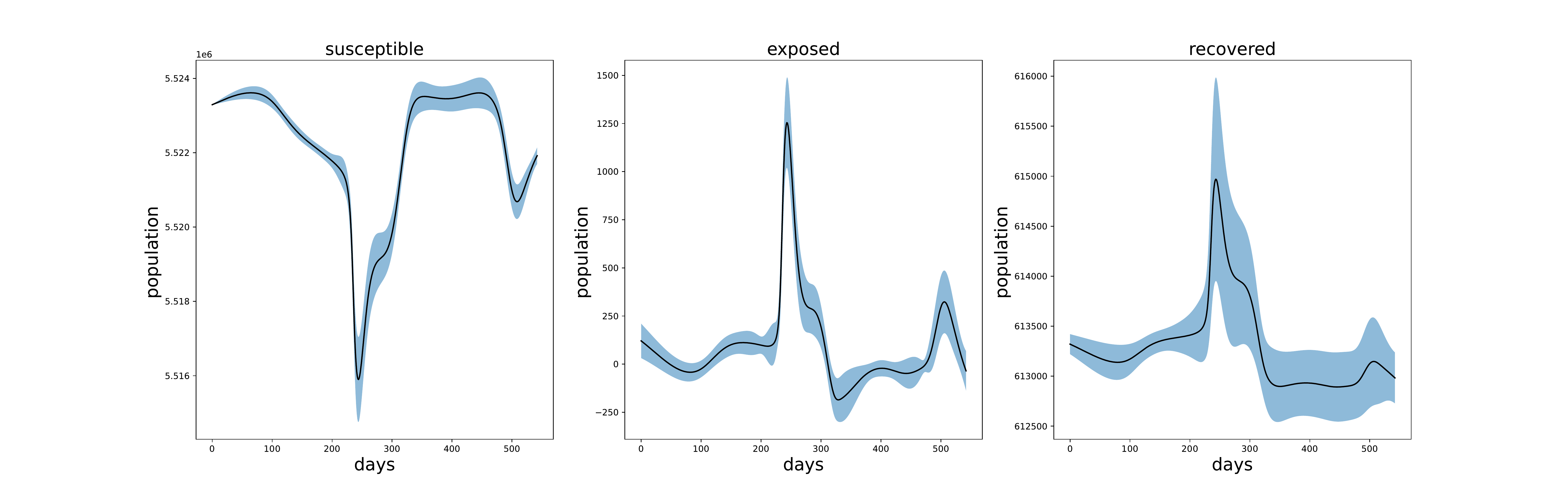}
\caption{learned Missouri Susceptible, Exposed, and Recovered daily population}
\label{missouri_4.1}
\end{figure}

\begin{figure}[htbp]
\centering
\includegraphics[width=17cm]{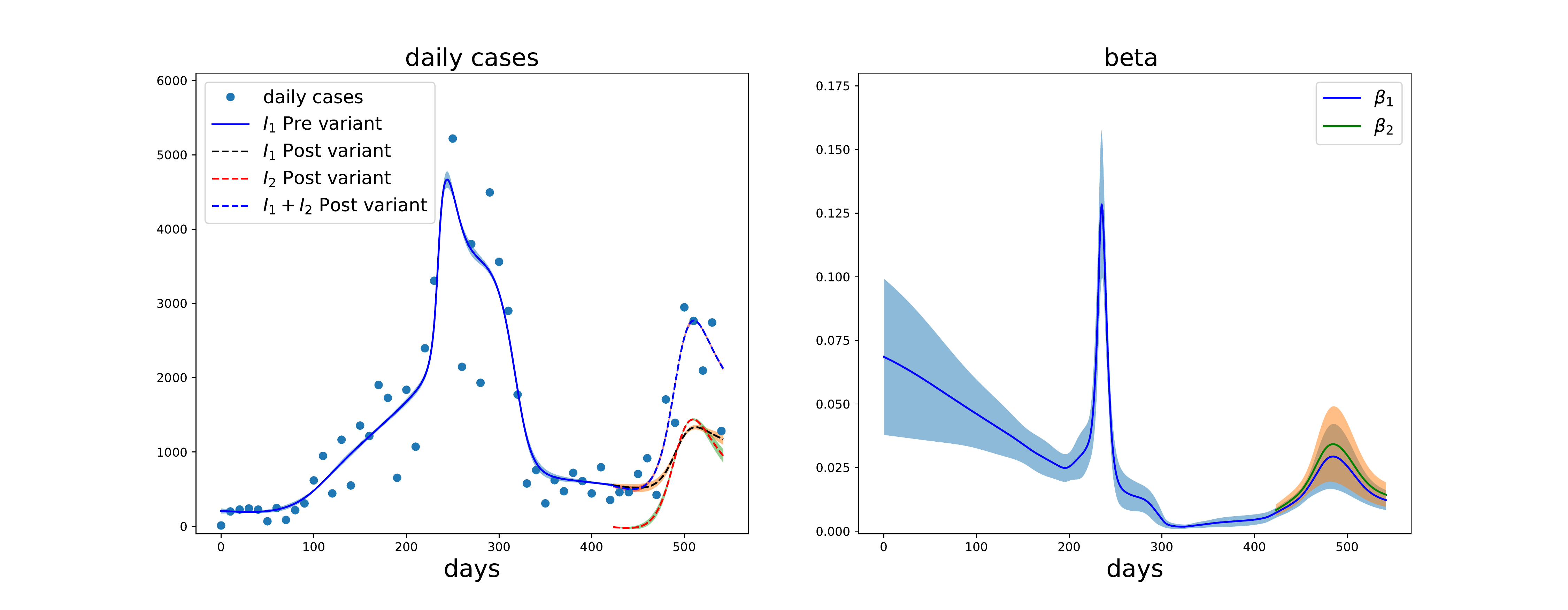}
\caption{Missouri daily cases and time-varying transmission rates}
\label{missouri}
\end{figure}


\begin{table}[htbp]
\centering
\begin{tabular}{lcc}
\addlinespace
\toprule
{Parameters} & {Mean} & {Std} \\
\toprule
pre- $\gamma_1$  & 0.01259 & 0.01111 \\
post- $\gamma_1$  & 0.00587 & 0.00821 \\
post- $\gamma_2$  & 0.00849 & 0.01081 \\
$\eta_1$  & 0.13721 & 0.03429 \\
$\eta_2$  & 0.19611 & 0.03427 \\
$\kappa_1$  & 1.04761 & 0.03035 \\
$\kappa_2$  & 1.13552 & 0.02867 \\
$(1 + \tau)$  & 1.09879 & 0.09738 \\
\midrule
\end{tabular}
\caption{Using Tennessee daily cases from March 2020 to September 2021, EINN Algorithm~\eqref{alg:Epidemiology informed neural network_var} learns the model parameters} 
\label{Table:tn1}
\end{table}

\begin{figure}[htbp]
\centering
\includegraphics[width=17cm]{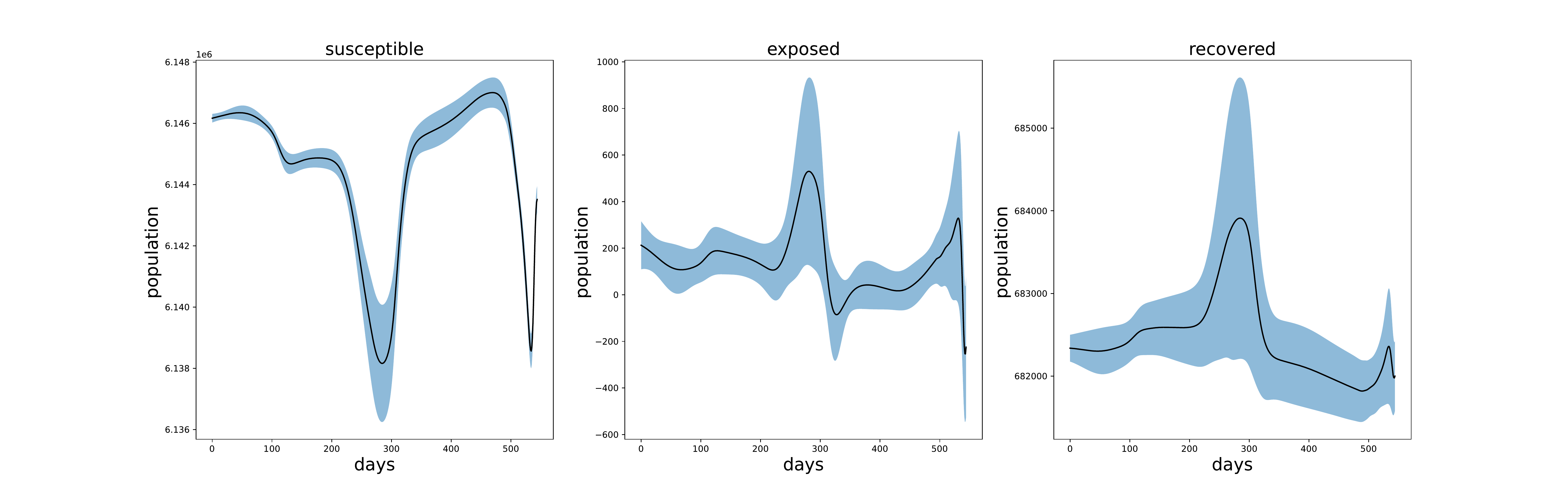}
\caption{learned Tennessee Susceptible, Exposed, and Recovered daily population}
\label{tennessee_4.1}
\end{figure}

\begin{figure}[htbp]
\centering
\includegraphics[width=17cm]{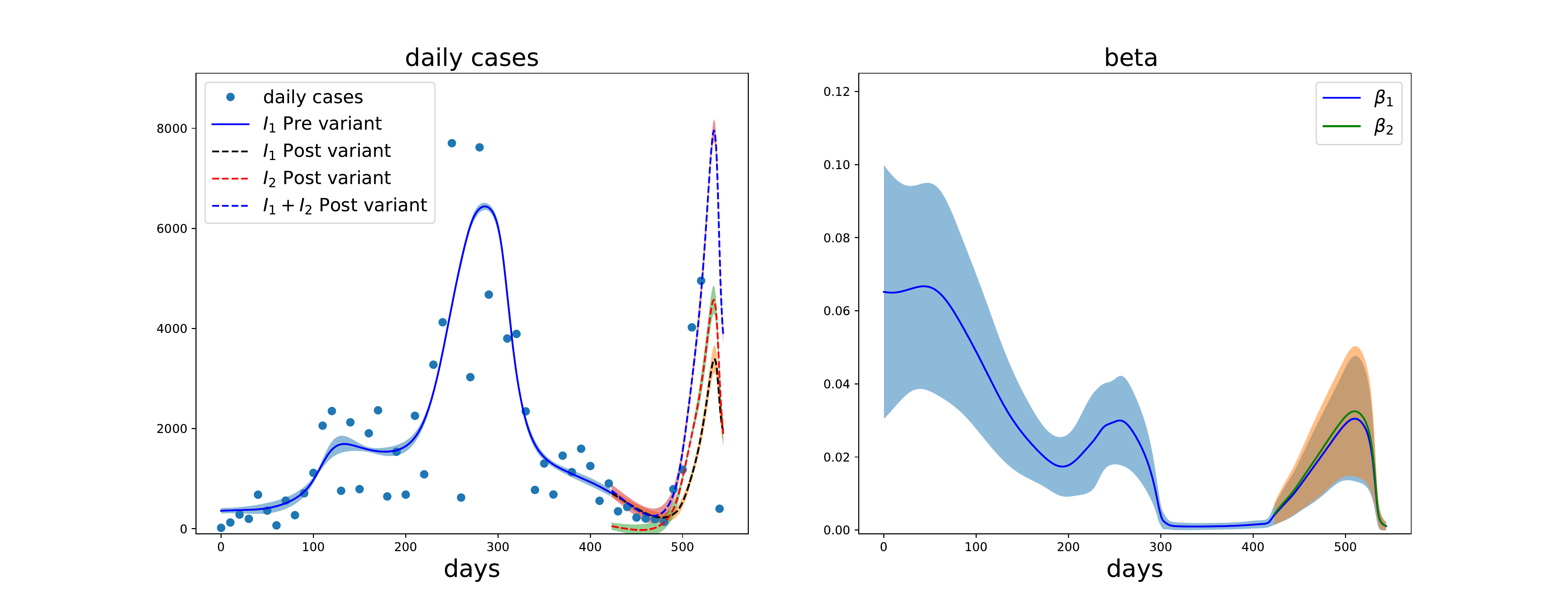}
\caption{Tennessee daily cases and time-varying transmission rates}
\label{tennessee}
\end{figure}


\begin{table}[htbp]
\centering
\begin{tabular}{lcc}
\addlinespace
\toprule
{Parameters} & {Mean} & {Std} \\
\toprule
pre- $\gamma_1$  & 0.02968 & 0.01594 \\
post- $\gamma_1$  & 0.00943 & 0.00985 \\
post- $\gamma_2$  & 0.00576 & 0.00516 \\
$\eta_1$  & 0.09304 & 0.06144 \\
$\eta_2$  & 0.24027 & 0.06143 \\
$\kappa_1$  & 1.03508 & 0.02477 \\
$\kappa_2$  & 1.13773 & 0.00892 \\
$(1 + \tau)$  & 1.12553 & 0.11431 \\
\midrule
\end{tabular}
\caption{Using Florida daily cases from March 2020 to September 2021, EINN Algorithm~\eqref{alg:Epidemiology informed neural network_var} learns the model parameters} 
\label{Table:fl1}
\end{table}

\begin{figure}[htbp]
\centering
\includegraphics[width=17cm]{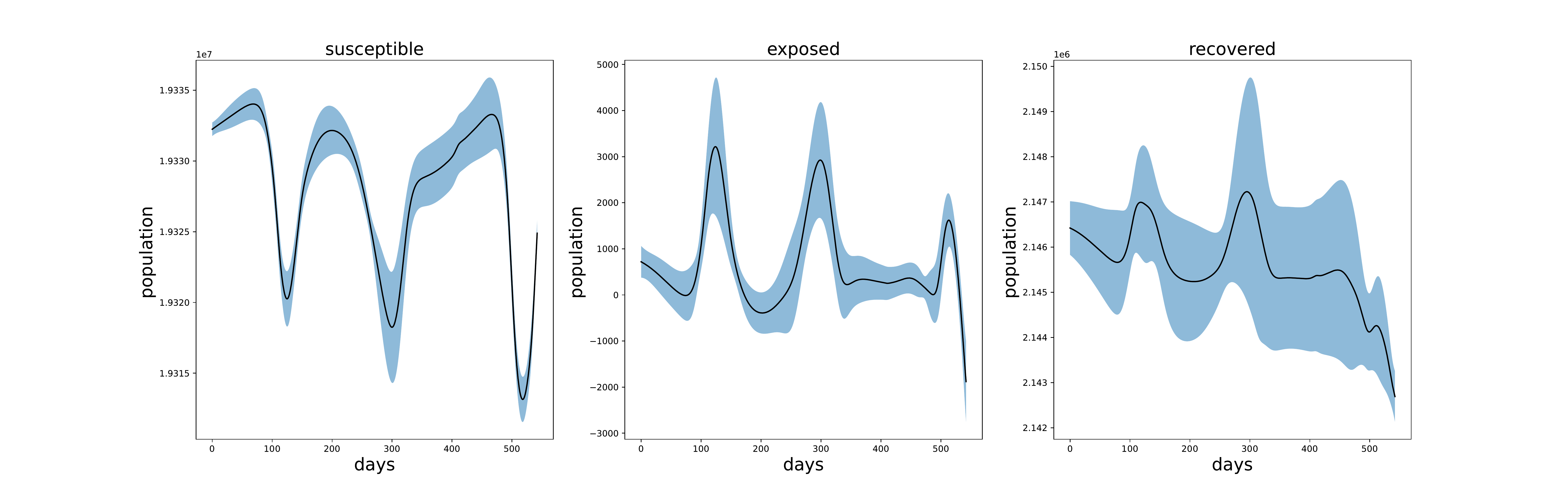}
\caption{learned Florida Susceptible, Exposed, and Recovered daily population}
\label{florida_4.1}
\end{figure}

\begin{figure}[htbp]
\centering
\includegraphics[width=17cm]{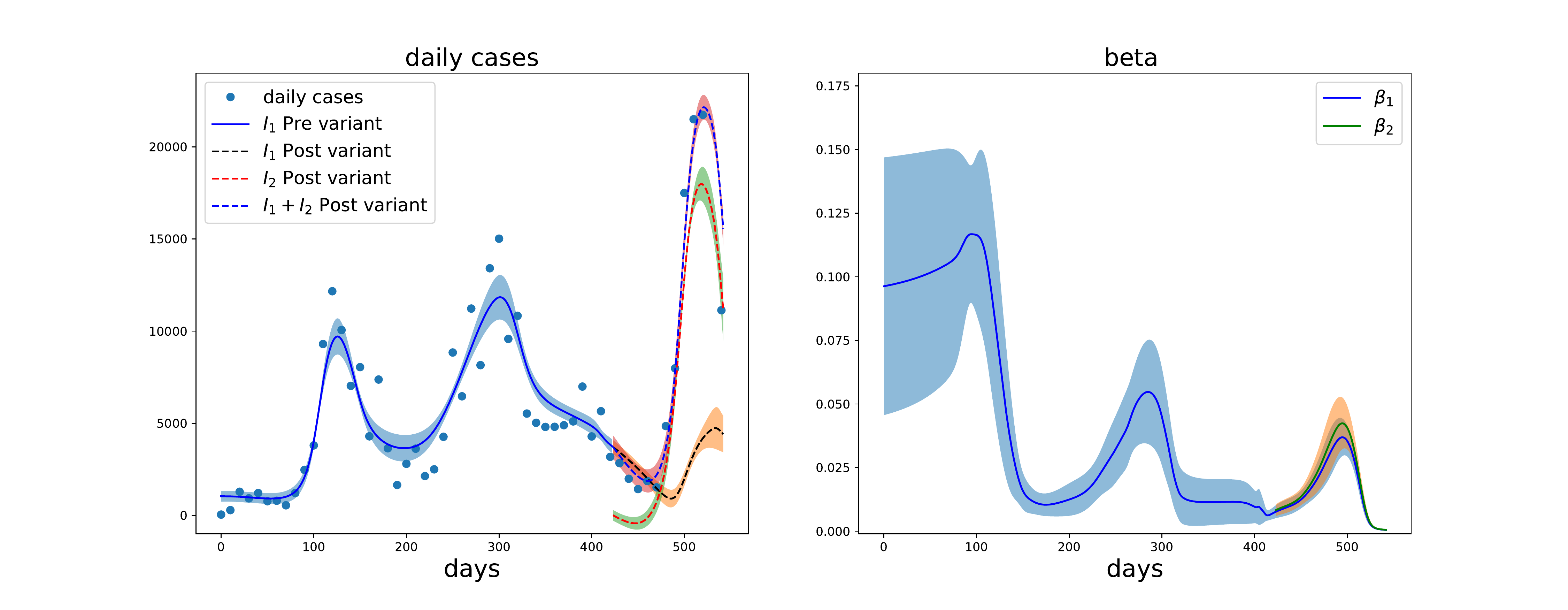}
\caption{Florida daily cases and time-varying transmission rates}
\label{florida}
\end{figure}

\section{Forecasting daily cases}

Forecasting the spread of infectious diseases in many studies are based on multiple linear regression (MLR), ordinary least squares regression (OLSR), principal component regression (PCR) and partial least squares regression (PLSR) and statistical methods such as the Auto Regressive Moving Average (ARIMA) and its many variants~\cite{Eftekhari2018,Du2018,Chimula2020}. These statistical methods are not optimal for nonlinear predictive task. This has motivated a shift towards techniques that rely on neural networks and neuro-fuzzy models~\cite{Hossain2020}. In this Section, we present an hybrid neural network that combines the simplicity and nonlinear learning capabilities of the Epidemiology-informed neural network (EINN) as well as the fuzzy inference system (ANFIS).

Adaptive neuro-fuzzy inference system (ANFIS), an hybrid neural network itself, is a combination of fuzzy logic and a feedforward neural network. It incorporates the advantages of both methods including learning capabilities, interpretability, quick convergence, adaptability and high accuracy. ANFIS displays excellent performance in approximation and prediction of nonlinear relationships in various fields \cite{Vacilopoulos2008}.

The Adaptive Neuro-Fuzzy Inference System (ANFIS) was introduced in  \cite{Jang1993}. It combines a neural network with a fuzzy inference system (FIS) based on ``IF-THEN'' rules. One major advantage of FIS is that it does not require knowledge of the main physical process as a pre-condition. ANFIS combines FIS with a backpropagation algorithm. These techniques provide a method for the fuzzy modeling procedure to learn from the available dataset, in order to compute the membership function parameters that best allow the  fuzzy inference system to track the given input/output data. 


To forecast the transmission of a multi-variant COVID-19, we present an efficient deep learning forecast model which combines two neural networks, we solve the ODE system using an Epidemiology Informed Neural Network (EINN) and we forecast using an adaptive neuro-fuzzy system (ANFIS), which we called the EINN-ANFIS model.

\section{Comparison of Forecasting techniques}
\label{data2}

We present results of the implementation of ANFIS, EINN-ANFIS, LSTM, EINN-LSTM for COVID-19 data from Alabama, Missouri, Tennessee, and Florida from March 2020 to September 2021. In the ANFIS approach, We used 4 regressors, 12 membership rules, and learning rate of 0.002. Training was done using 300 epochs, where we used the adams optimizer and for the loss function, we used the mean square error. The EINN-ANFIS is a hybrid neural network, where EINN is first used to train the daily cases dataset and a second round of training is done using ANFIS. In the LSTM approach, we used 4 input layers which corresponds to the daily cases at times $t$, $t+1$, $t+2$, and $t+3$. The adams optimizer is also used in training the LSTM with 20 epochs and the loss function also uses the mean square error. In the EINN-LSTM approach, a first batch of training is done using the EINN algorithm and then a second batch of training is done using LSTM. In Tables~\eqref{Table:al2}--\eqref{Table:fl2}, we present the validation loss of each method.

\begin{table}[htbp]
\centering
\begin{tabular}{lcc}
\addlinespace
\toprule
{Method} & {Mean} & {Std} \\
\toprule
ANFIS  & 0.00048 & 0.00098 \\
EINN-ANFIS  & 0.00032 & 0.00050 \\
LSTM  & 0.00141 & 0.00004 \\
LSTM-EINN  & 0.00110 & 0.00006 \\
\midrule
\end{tabular}
\caption{Validation loss in the ANFIS, EINN-ANFIS, LSTM, and LSTM-EINN forecasting technique for Alabama daily cases from March 2020 to September 2021.}
\label{Table:al2}
\end{table}

\begin{figure}[htbp]
    \subfloat[ANFIS]{\includegraphics[width = 3.2in]{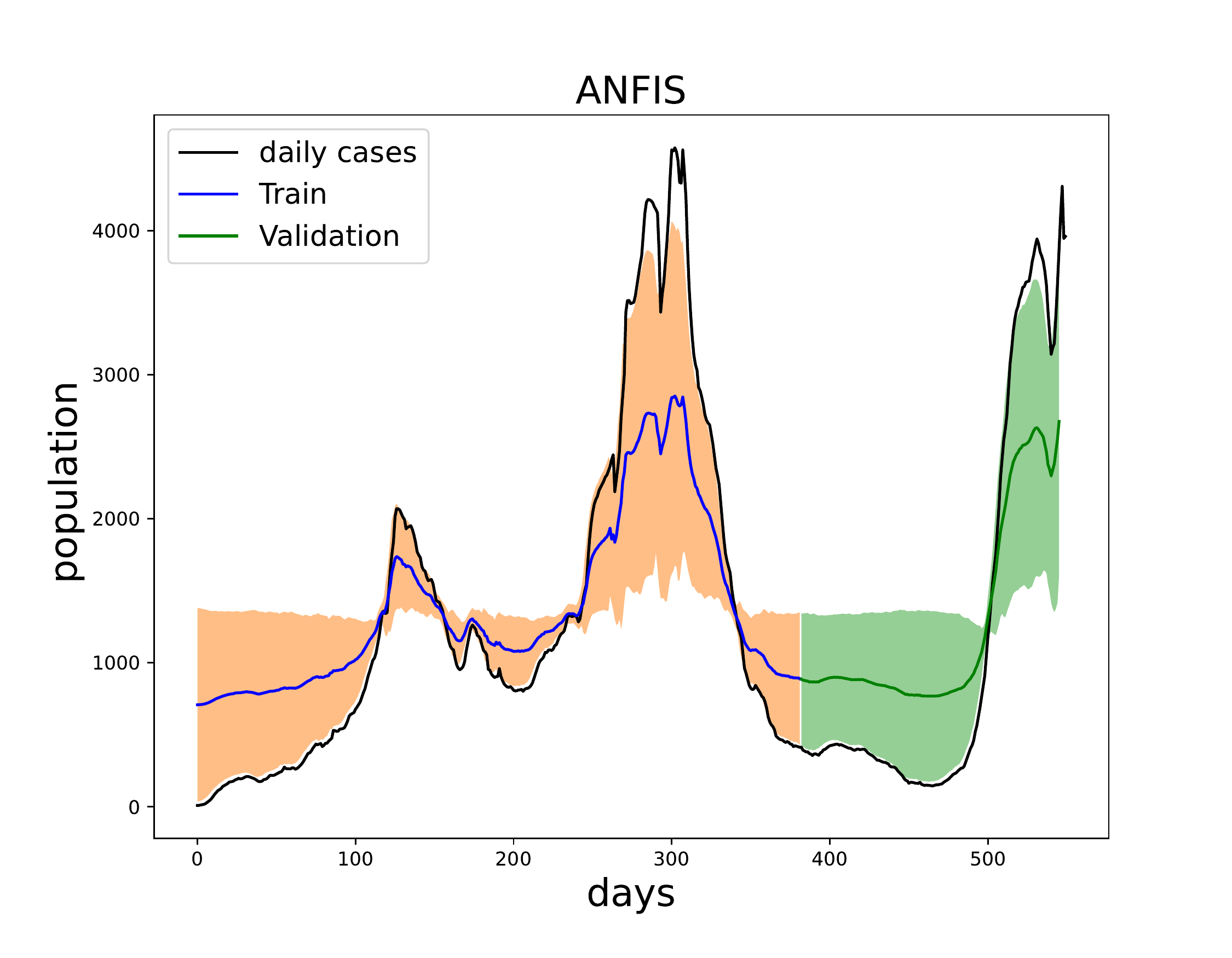}} 
    \hfill
    \subfloat[EINN-ANFIS]{\includegraphics[width = 3.2in]{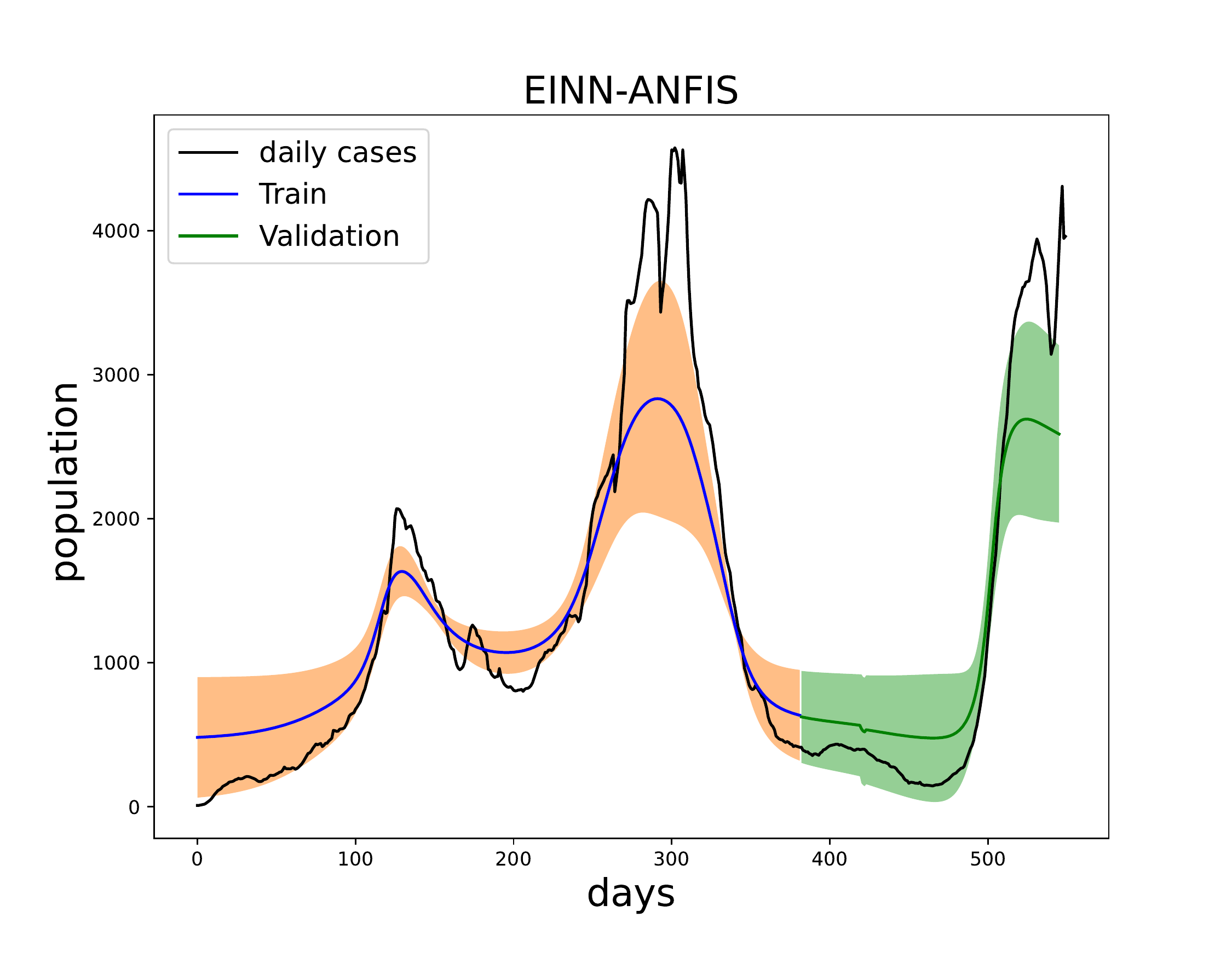}}
    \vskip .4in
    \subfloat[LSTM]{\includegraphics[width = 3.2in]{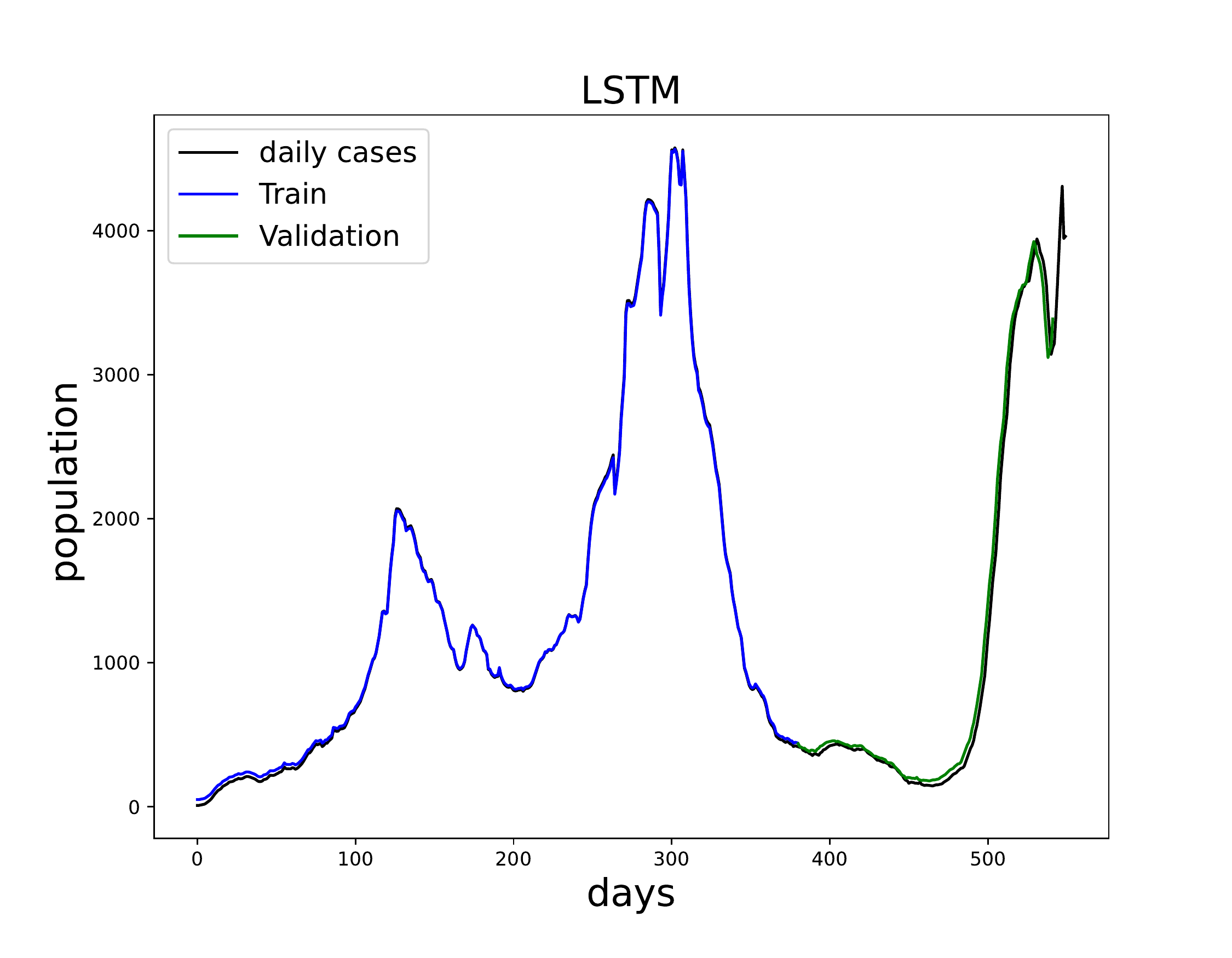}}
    \hfill
    \subfloat[EINN-LSTM]{\includegraphics[width = 3.2in]{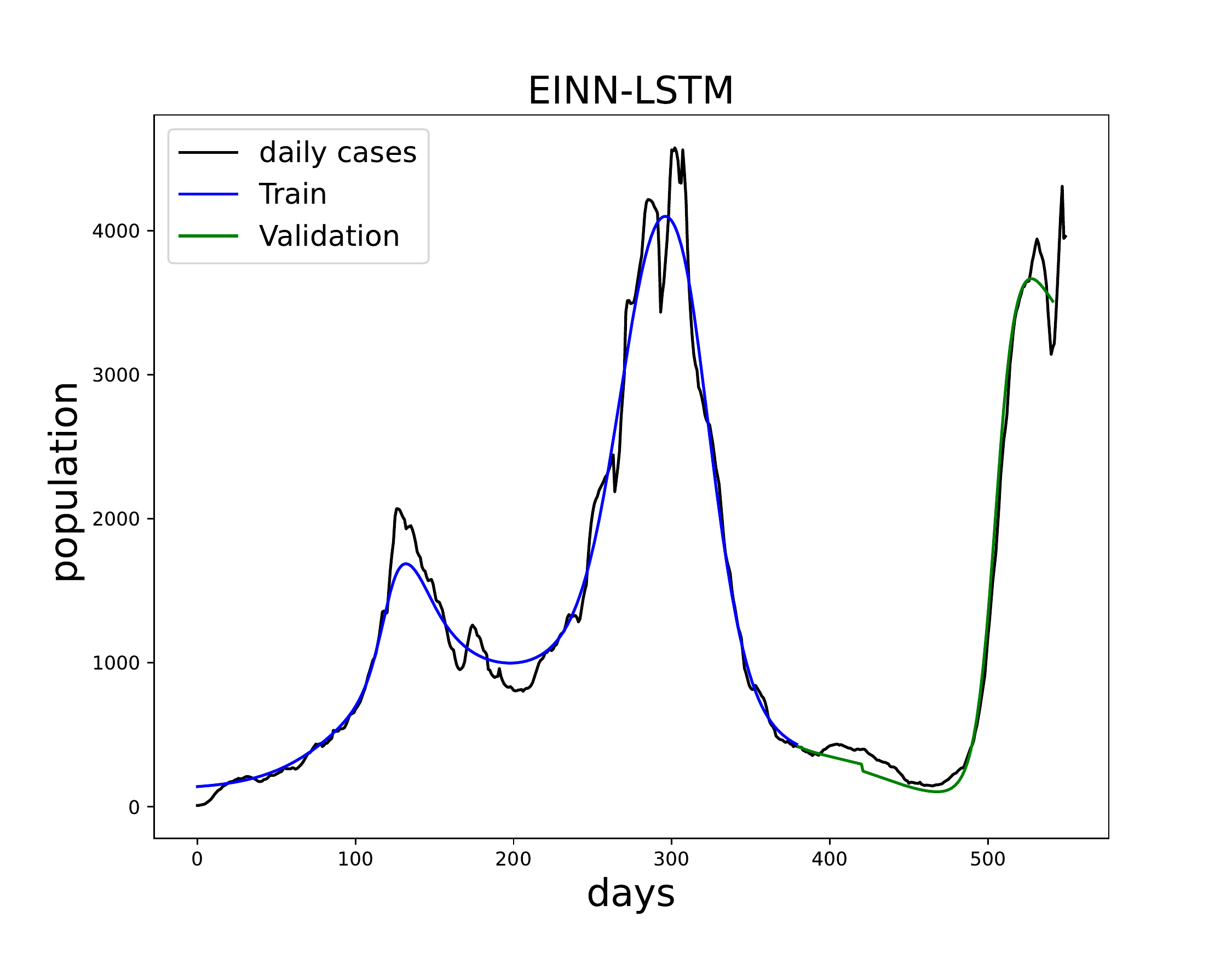}} 
\caption{Alabama daily cases forecasting using ANFIS, EINN-ANFIS, LSTM, LSTM-EINN}
\label{alabama-anfis}
\end{figure}


\begin{table}[htbp]
\centering
\begin{tabular}{lcc}
\addlinespace
\toprule
{Method} & {Mean} & {Std} \\
\toprule
ANFIS  & 0.00011 & 0.00019 \\
EINN-ANFIS  & 0.00004 & 0.00006 \\
LSTM  & 0.00333 & 0.00003 \\
LSTM-EINN  & 0.00118 & 0.00012 \\
\midrule
\end{tabular}
\caption{Validation loss in the ANFIS, EINN-ANFIS, LSTM, and LSTM-EINN forecasting technique for Missouri daily cases from March 2020 to September 2021.}
\label{Table:mo2}
\end{table}

\begin{figure}[htbp]
    \subfloat[ANFIS]{\includegraphics[width = 3.2in]{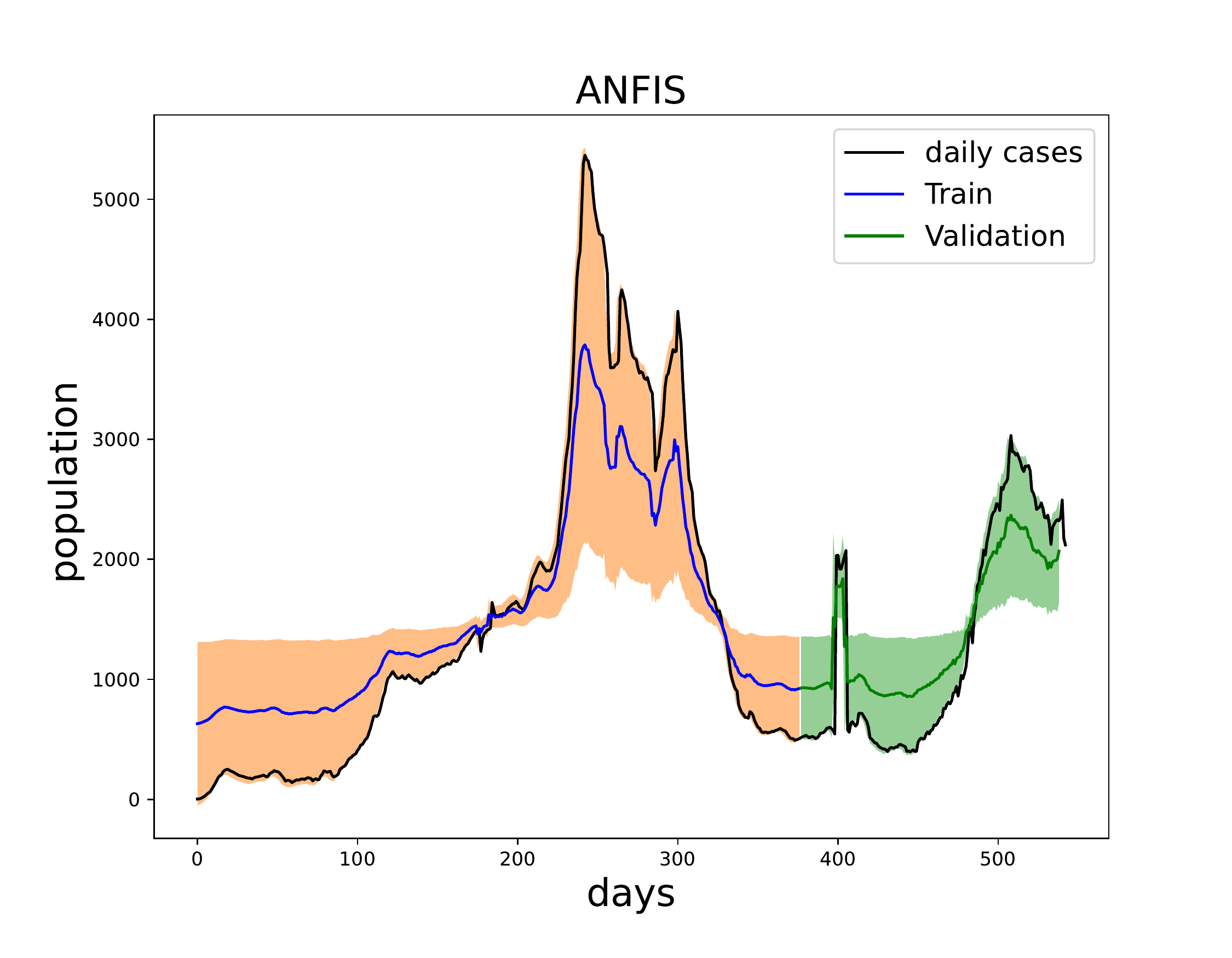}} 
    \hfill
    \subfloat[EINN-ANFIS]{\includegraphics[width = 3.2in]{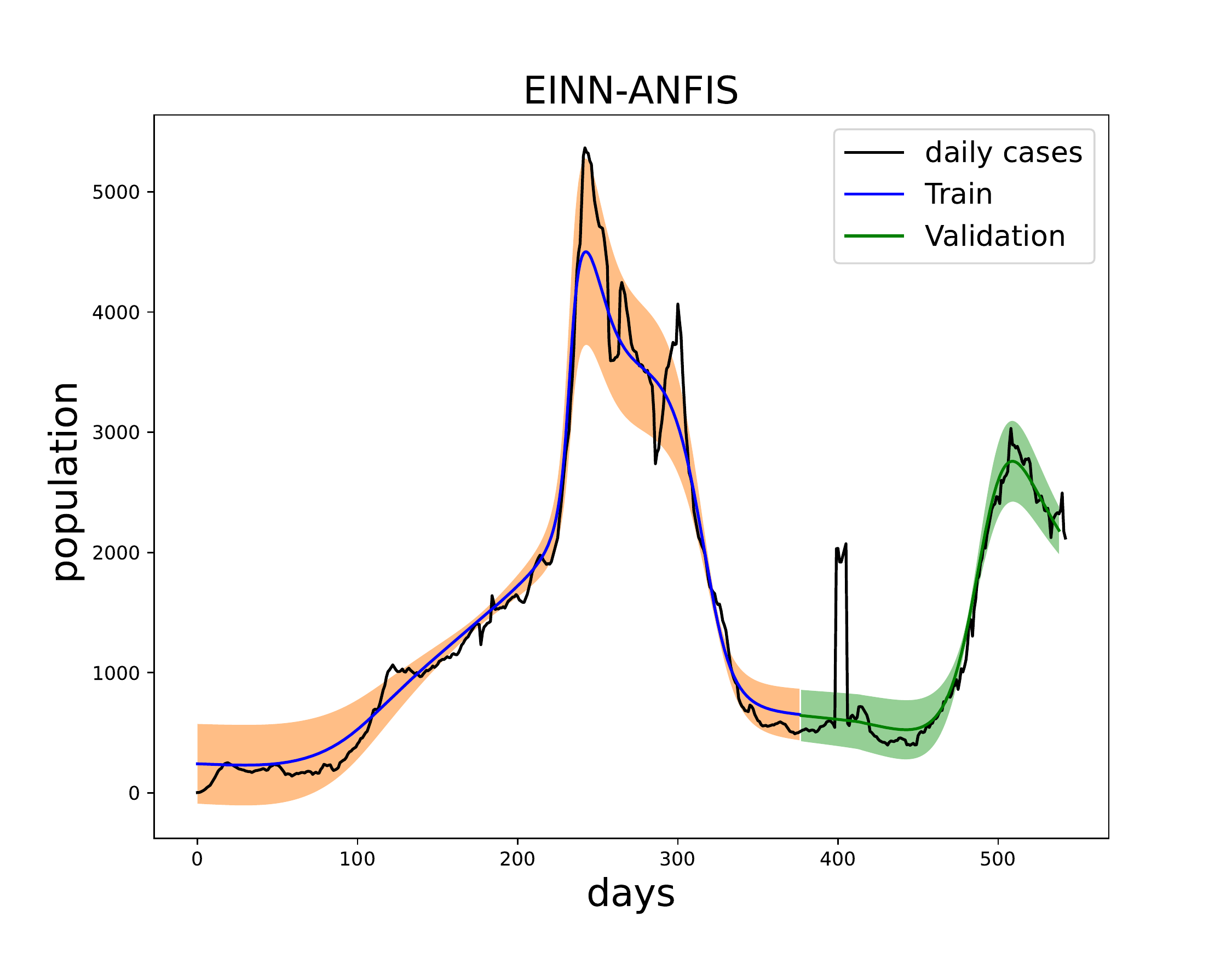}}
    \vskip .4in
    \subfloat[LSTM]{\includegraphics[width = 3.2in]{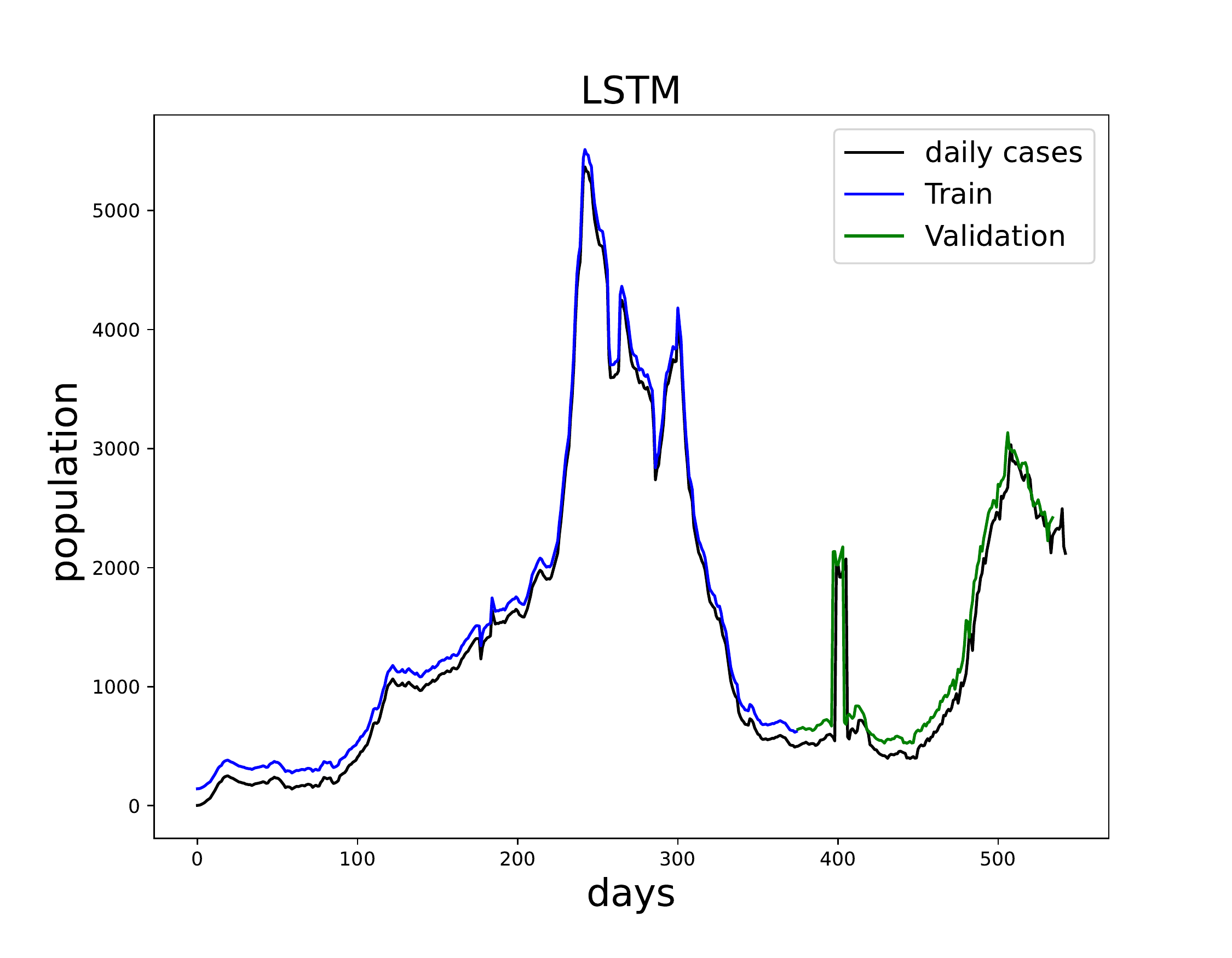}}
    \hfill
    \subfloat[EINN-LSTM]{\includegraphics[width = 3.2in]{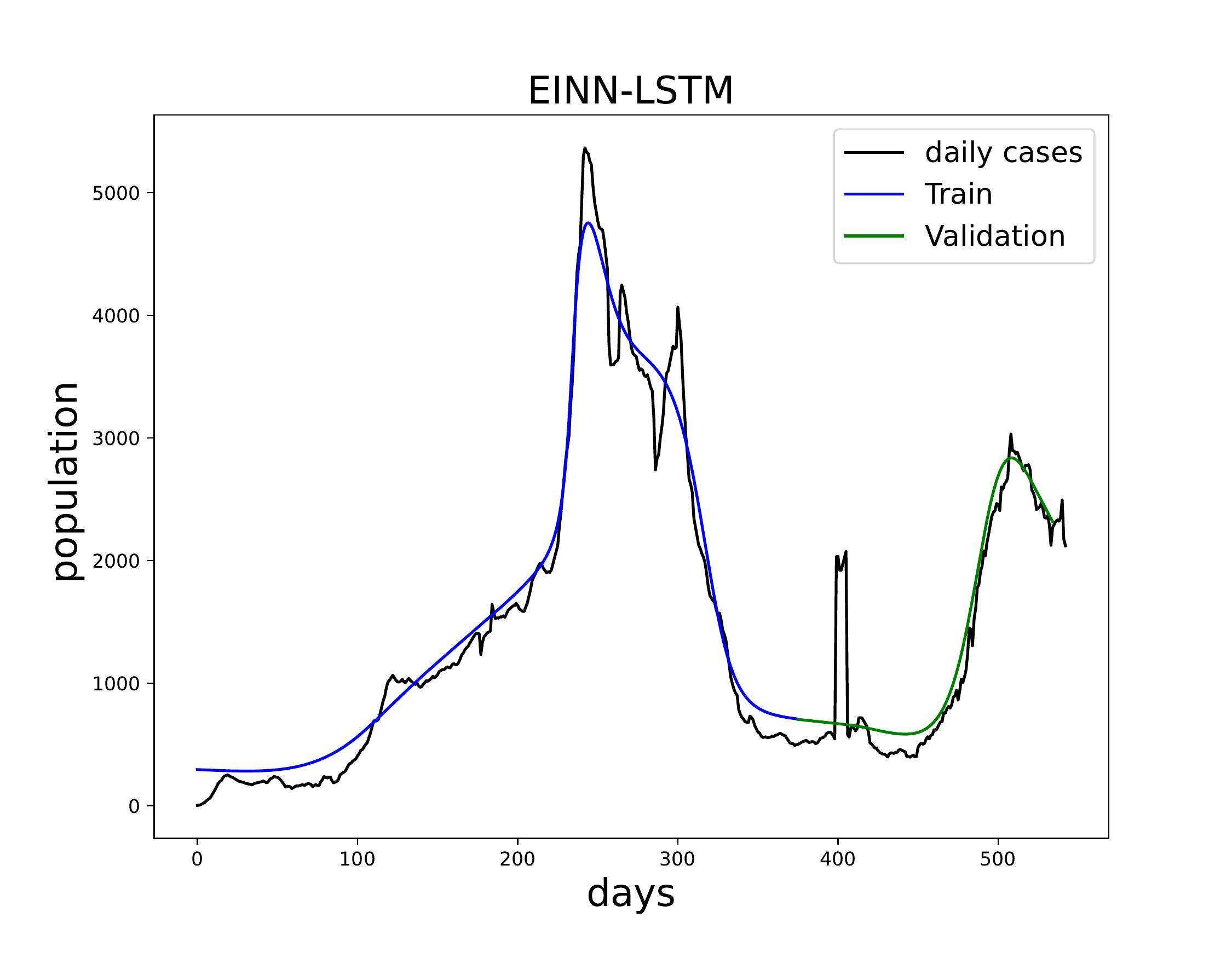}} 
\caption{Missouri daily cases forecasting using ANFIS, EINN-ANFIS, LSTM, LSTM-EINN}
\label{missouri-anfis}
\end{figure}


\begin{table}[htbp]
\centering
\begin{tabular}{lcc}
\addlinespace
\toprule
{Method} & {Mean} & {Std} \\
\toprule
ANFIS  & 0.00061 & 0.00125 \\
EINN-ANFIS  & 0.00033 & 0.00056 \\
LSTM  & 0.00267 & 0.00011 \\
LSTM-EINN  & 0.00183 & 0.00009 \\
\midrule
\end{tabular}
\caption{Validation loss in the ANFIS, EINN-ANFIS, LSTM, and LSTM-EINN forecasting technique for Tennessee daily cases from March 2020 to September 2021.}
\label{Table:tn2}
\end{table}

\begin{figure}[htbp]
    \subfloat[ANFIS]{\includegraphics[width = 3.2in]{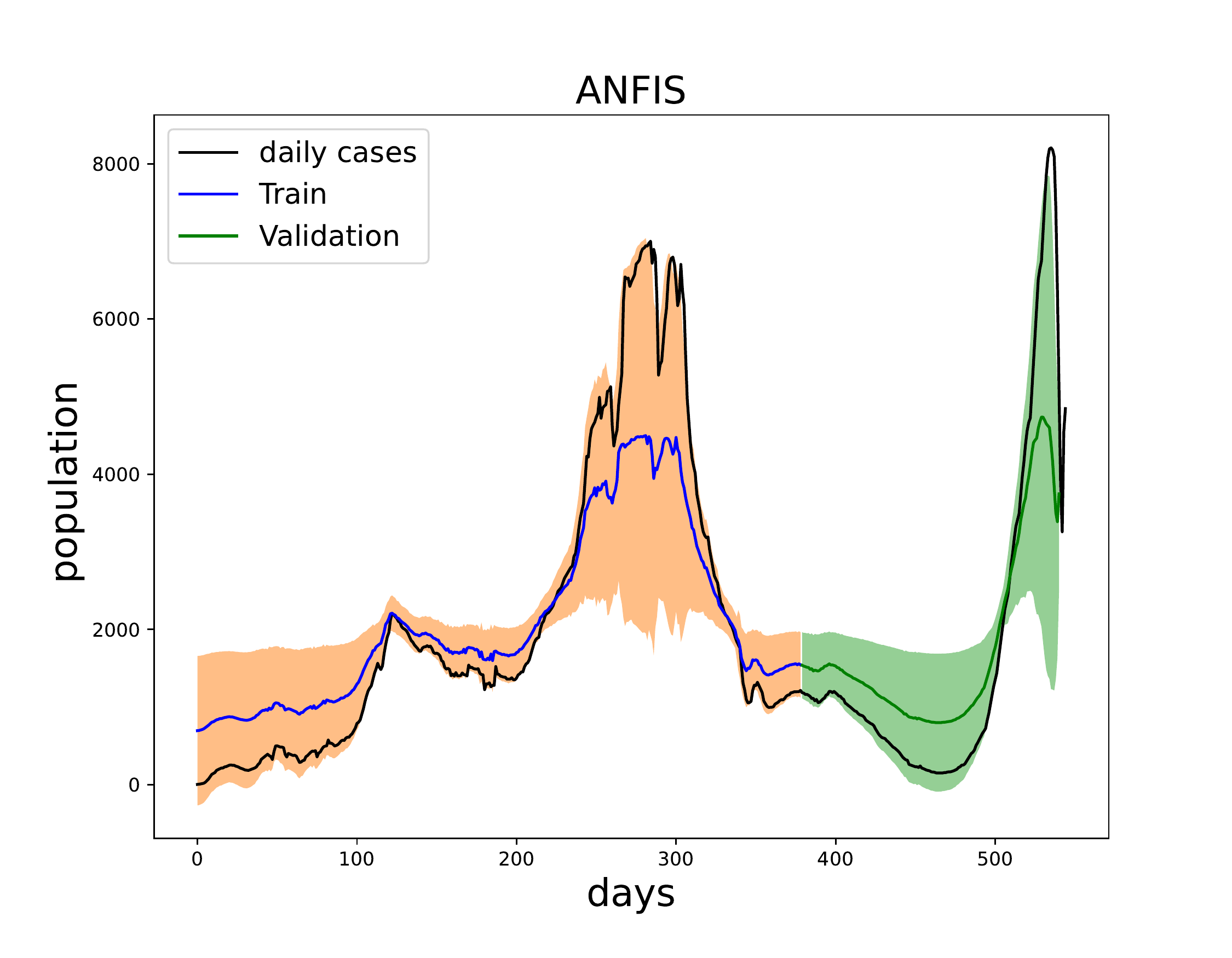}} 
    \hfill
    \subfloat[EINN-ANFIS]{\includegraphics[width = 3.2in]{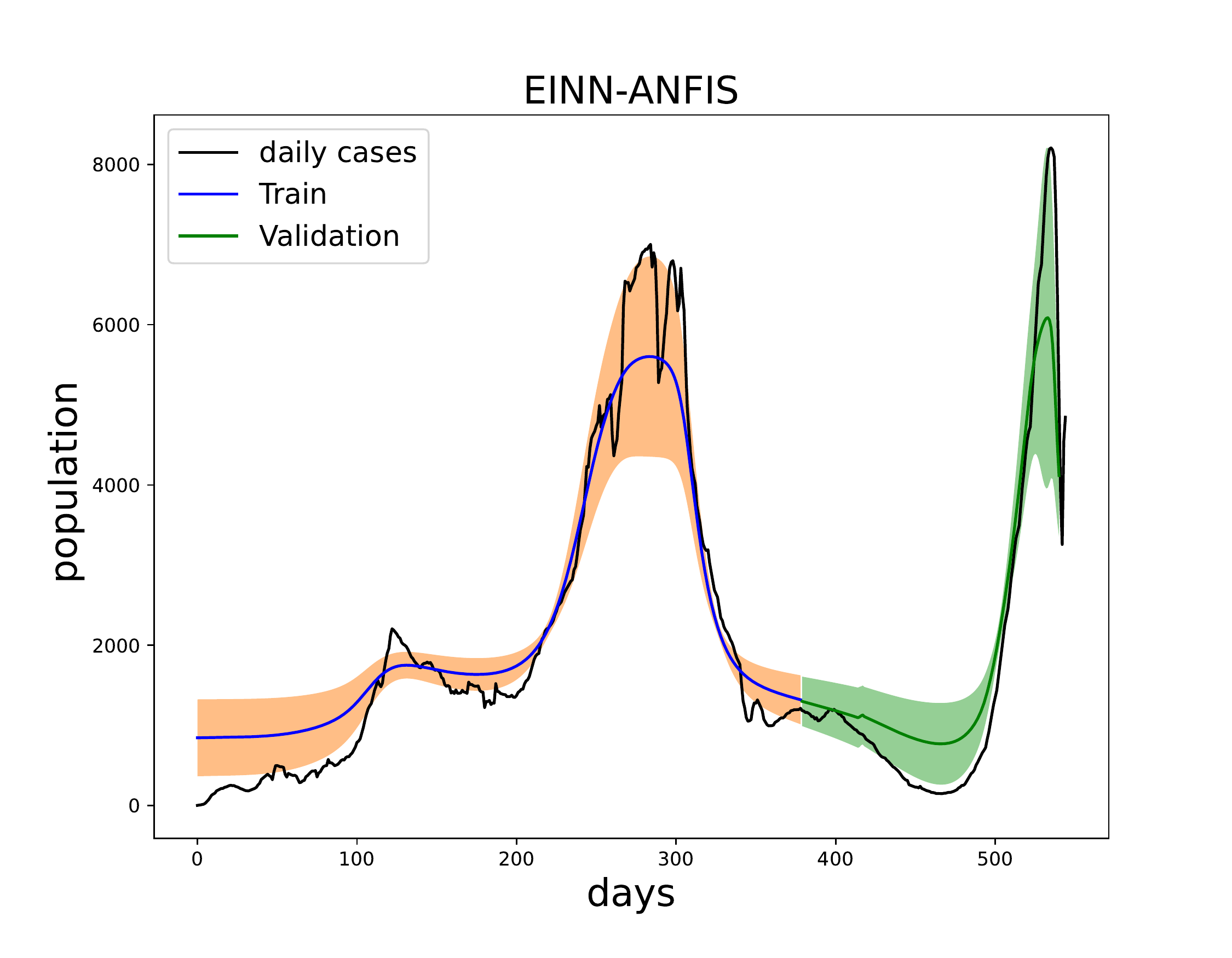}}
    \vskip .4in
    \subfloat[LSTM]{\includegraphics[width = 3.2in]{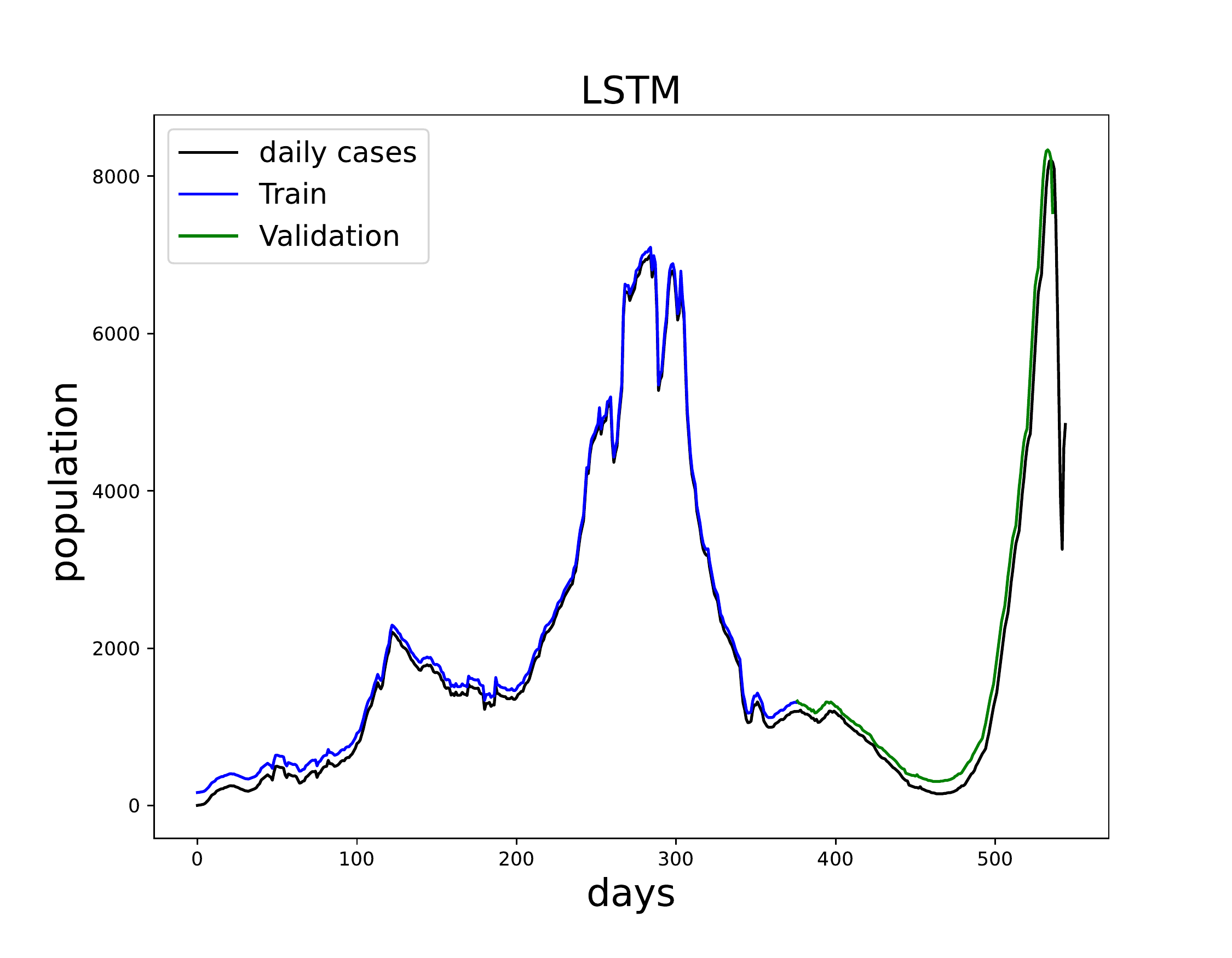}}
    \hfill
    \subfloat[EINN-LSTM]{\includegraphics[width = 3.2in]{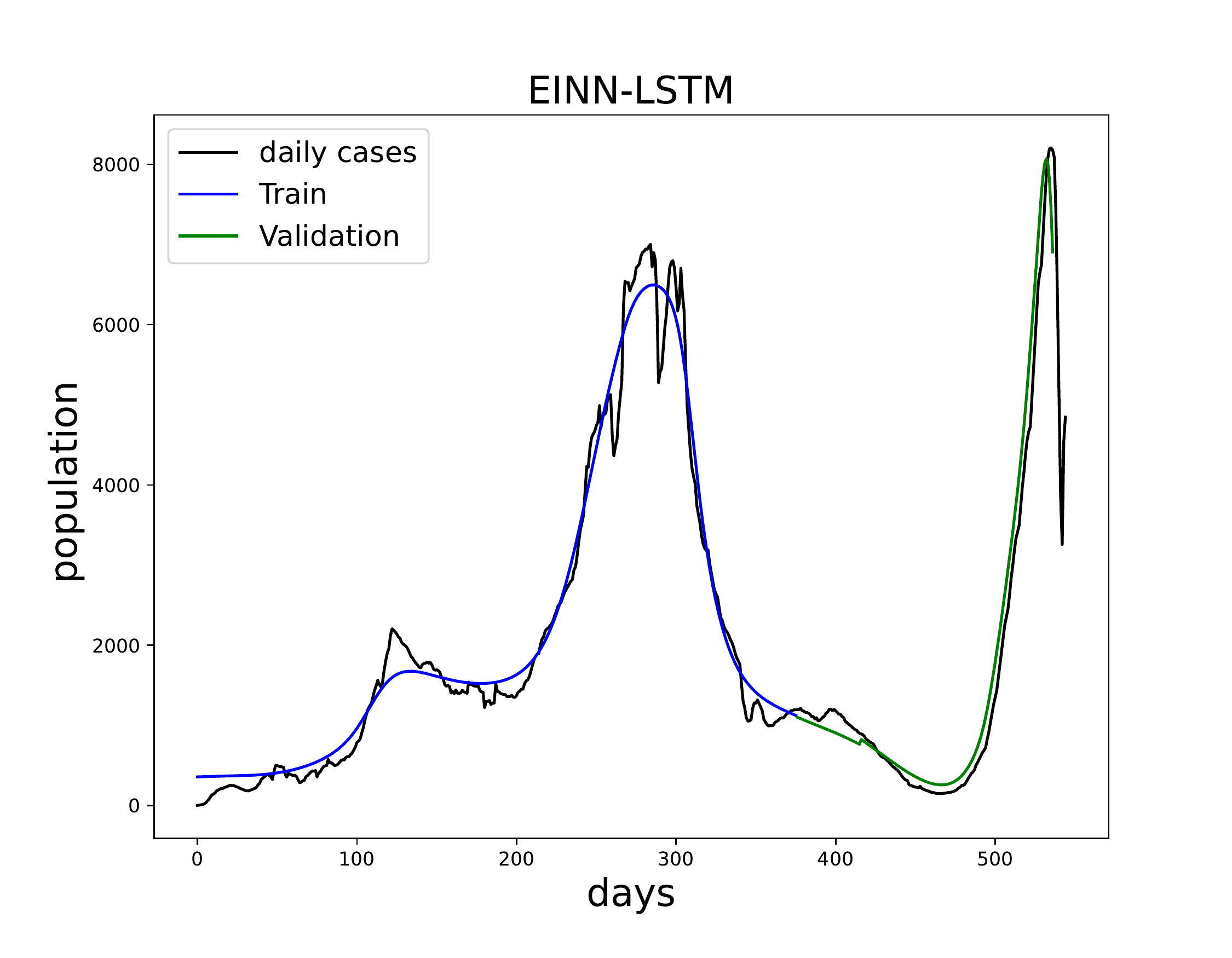}} 
\caption{Tennessee daily cases forecasting using ANFIS, EINN-ANFIS, LSTM, LSTM-EINN}
\label{tennessee-anfis}
\end{figure}


\begin{table}[htbp]
\centering
\begin{tabular}{lcc}
\addlinespace
\toprule
{Method} & {Mean} & {Std} \\
\toprule
ANFIS  & 0.00199 & 0.00284 \\
EINN-ANFIS  & 0.00249 & 0.00347 \\
LSTM  & 0.00169 & 0.00009 \\
LSTM-EINN  & 0.00149 & 0.00014 \\
\midrule
\end{tabular}
\caption{Validation loss in the ANFIS, EINN-ANFIS, LSTM, and LSTM-EINN forecasting technique for Florida daily cases from March 2020 to September 2021.}
\label{Table:fl2}
\end{table}

\begin{figure}[htbp]
    \subfloat[ANFIS]{\includegraphics[width = 3.2in]{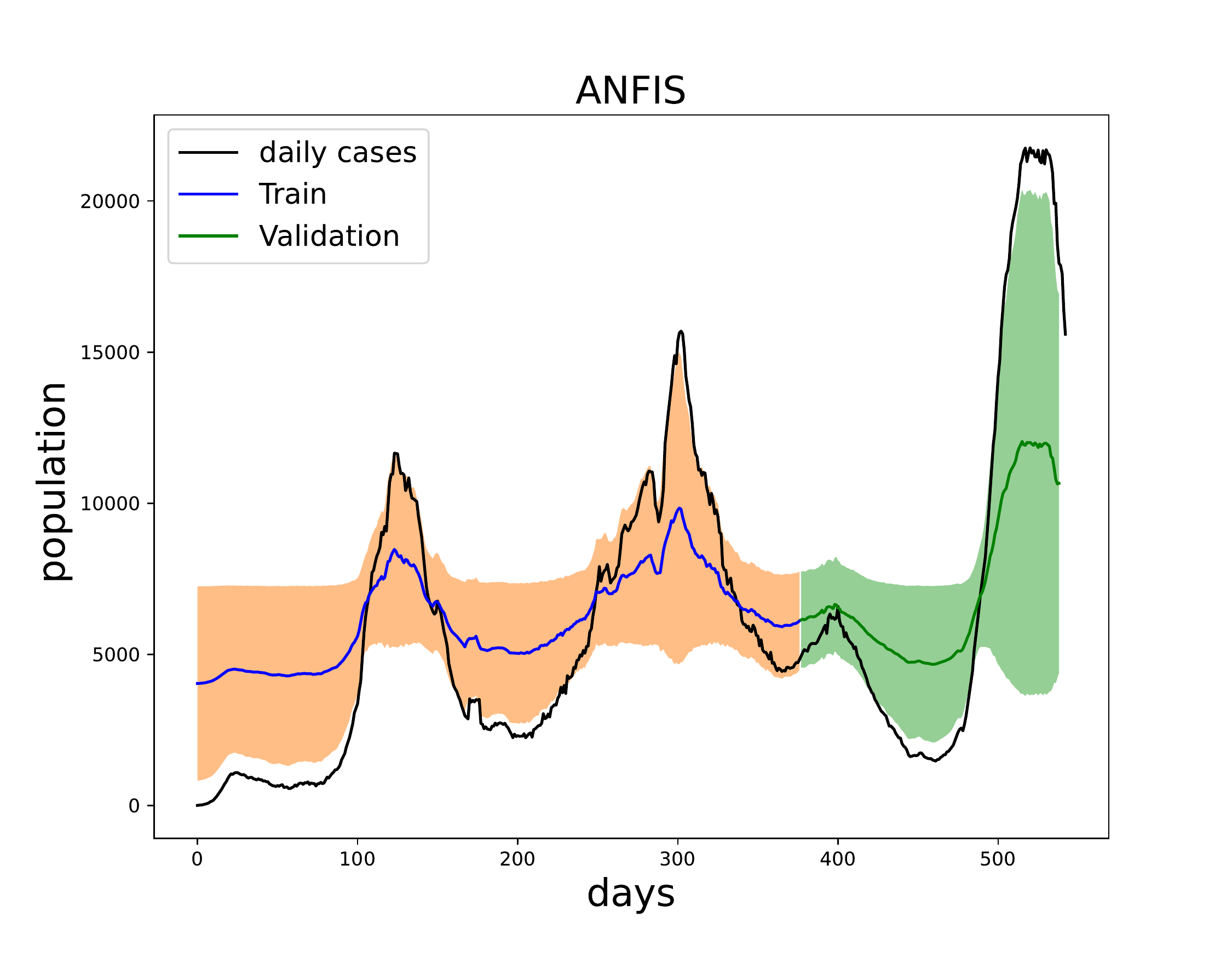}} 
    \hfill
    \subfloat[EINN-ANFIS]{\includegraphics[width = 3.2in]{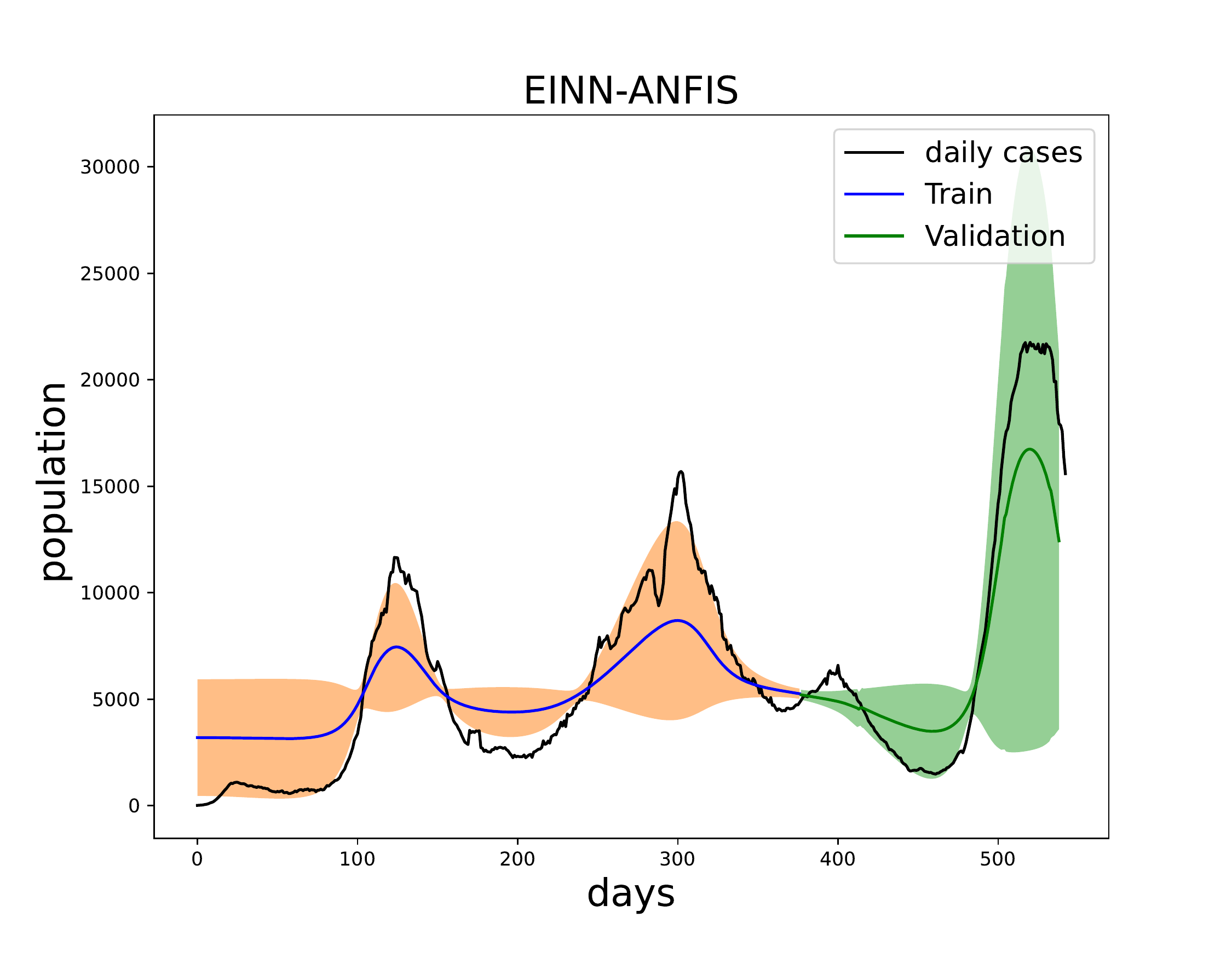}}
    \vskip .4in
    \subfloat[LSTM]{\includegraphics[width = 3.2in]{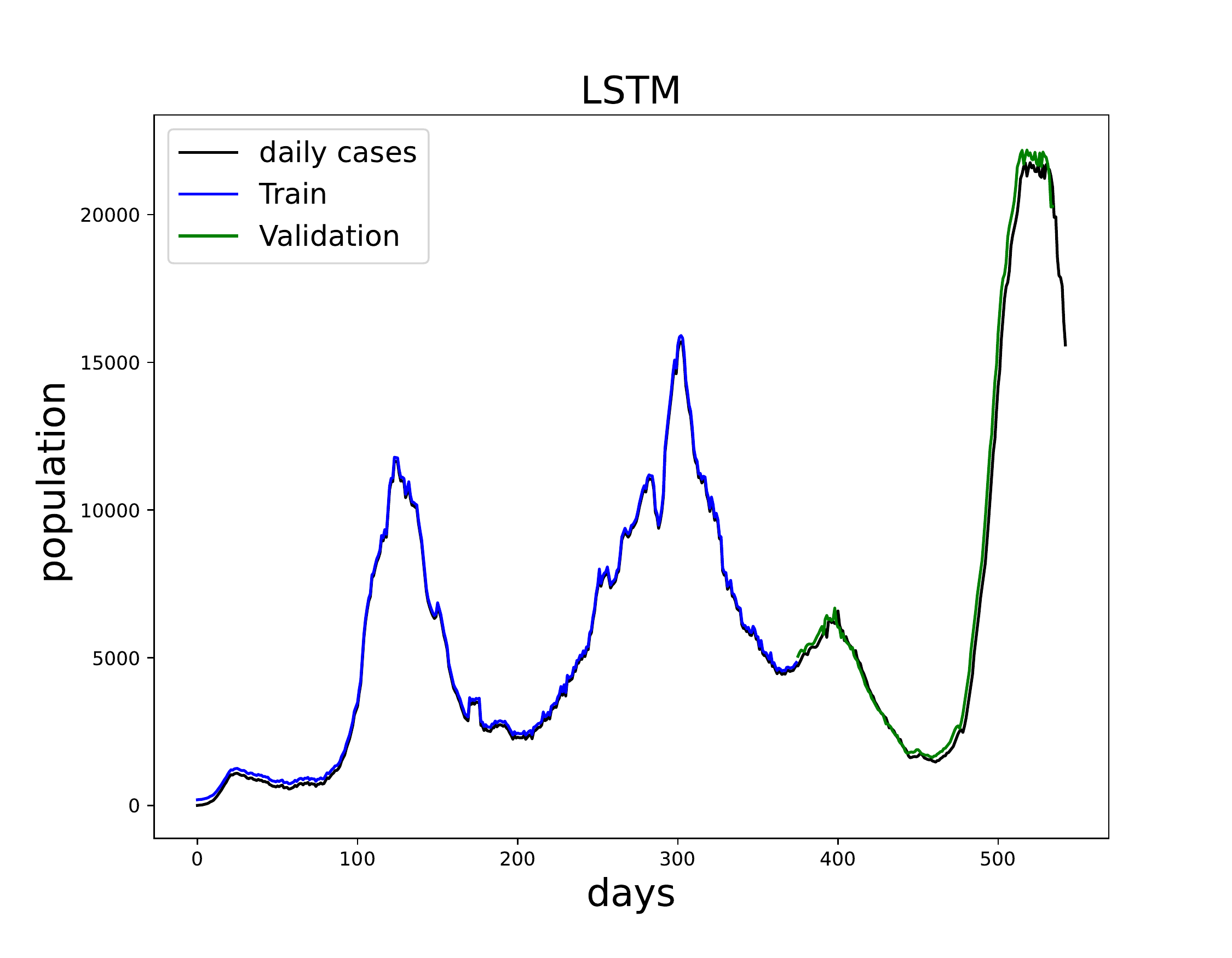}}
    \hfill
    \subfloat[EINN-LSTM]{\includegraphics[width = 3.2in]{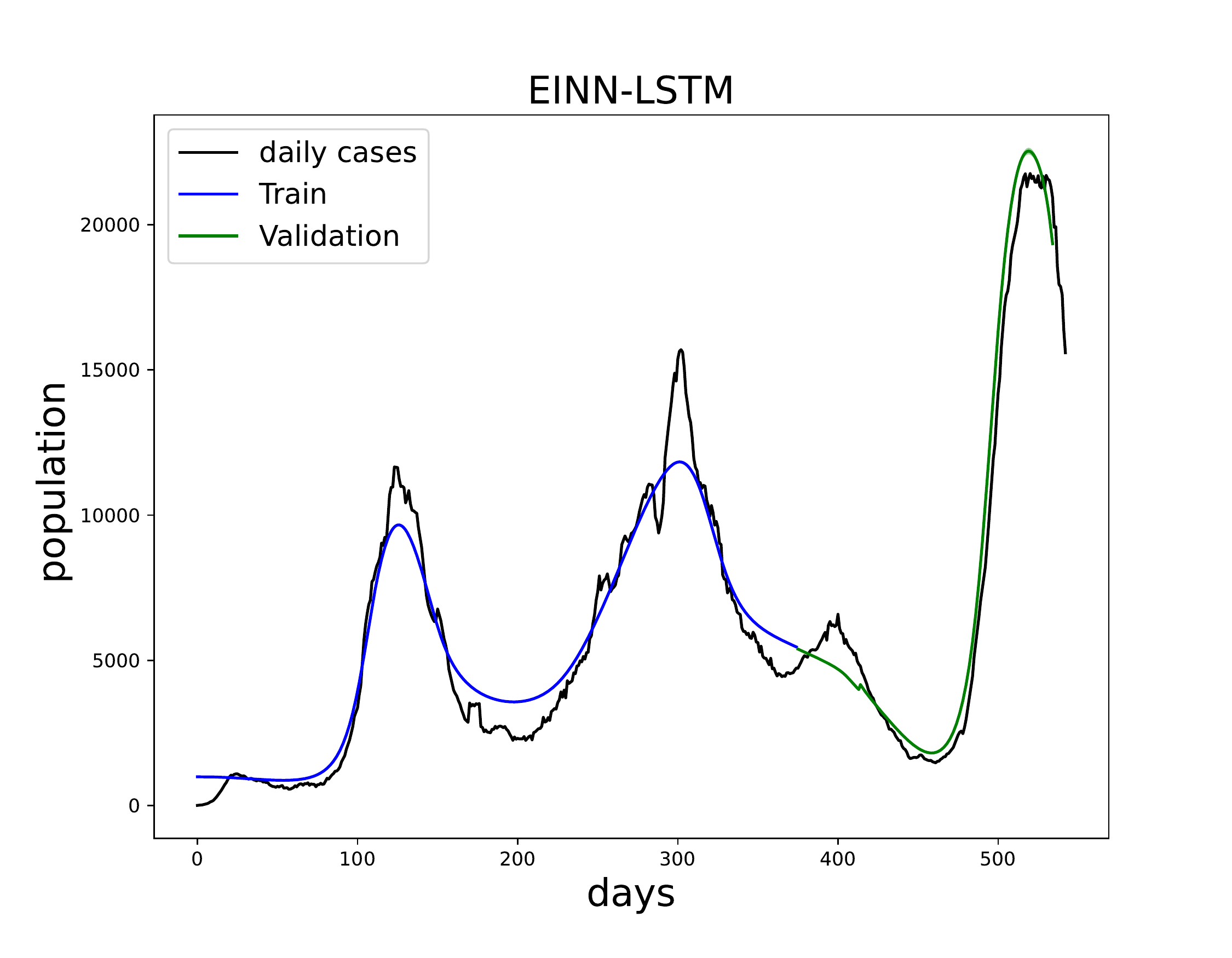}} 
\caption{Florida daily cases forecasting using ANFIS, EINN-ANFIS, LSTM, LSTM-EINN}
\label{florida-anfis}
\end{figure}

As can be observed from these Tables~\eqref{Table:al2}--\eqref{Table:fl2} EINN-ANFIS is an improvement over ANFIS and similarly, EINN-LSTM is an improvement over LSTM.

\section{Performance analysis of error metrics}
\label{metrics}

The following error metrics are used in our data driven simulation:

\begin{itemize}
\item Root Mean Square Error (RMSE):\\
\[RMSE = \sqrt{\frac{1}{N_s}\sum_{i=1}^{N_s}(Y_i - \tilde{Y}_i)^2}\],\\
where $Y$ and $\tilde{Y}$ are the predicted and original values, respectively. \\

\item Mean Absolute Error (MAE):\\
\[MAE = \frac{1}{N_s}\sum_{i=1}^{N_s}|Y_i - \tilde{Y}_i|\].\\

\item Mean Absolute Percentage Error (MAPE):\\
\[MAPE = \frac{1}{N_s}\sum_{i=1}^{N_s}|\frac{Y_i - \tilde{Y}_i}{Y_i}|\].\\

\item Root Mean Squared Relative Error (RMSRE):\\
\[RMSRE = \sqrt{\frac{1}{N_s}\sum_{i=1}^{N_s}(\frac{Y_i - \tilde{Y}_i}{Y_i})^2}\],\\
$N_s$ represents the sample size of the data.
\end{itemize}

In Table~\ref{table2} We provide a comparison of error metrics for EINN using random splits for the training and test data.

\begin{table}[H]
\begin{tabular}{ |p{1.7cm}|p{2.0cm}| p{2.0cm}| p{1.7cm}| p{1.7cm}|}
\hline
$State$ & $RMSE$ & $MAE$ & $MAPE$ & $RMSRE$ \\
\hline
 Florida  & $0.00768642$ & $0.0868862$ & $0.62593067$ & $1.4576211$ \\
 Tennessee  & $0.01003919$ & $0.07956579$ & $1.41574097$ & $3.11372066$ \\
 Alabama  & $0.00772795$ & $0.08293544$ & $0.76998496$ & $0.24813209$ \\
 Missouri  & $0.0083841$ & $0.09985308$ & $1.65703082$ & $0.44569263$  \\
\hline
\end{tabular}
\caption{Error metrics for random split\label{table2}}
\end{table}

\section{Conclusion}
\label{conclusion}

We have presented a data-driven deep learning algorithm that learns time-varying transmission rates of multi-variant in an infectious disease such as COVID-19.
The algorithm we presented learns the nonlinear time-varying transmission rates without a pre-assumed pattern as well as predict the daily cases and daily recovered populations. 
We learn these population groups using only daily cases data. 
This approach is found useful when the dynamics of an epidemiological model such as an SEIR model is impacted by various mitigation measures. 
The algorithm presented in this paper can be adapted to most epidemiology models. 
Using US daily cases data, we demonstrate that the algorithm presented in this work can be combined together with recurrent neural networks and ANFIS for an improved short-term forecast.
This study is seen useful in the event of a pandemic such as COVID-19, where public health interventions and public response and perceptions interfere in the interaction of the compartments in an epidemiology model.

The computer codes will be available at \url{https://github.com/okayode/EINN-COVID}.

\bibliography{mybibfile}

\end{document}